\documentclass[letterpaper,twocolumn,10pt]{article}
\usepackage{usenix}

\usepackage{tikz}
\usepackage{amsmath}

\usepackage{filecontents}
\usepackage{todonotes}

\usepackage{amsmath}
\usepackage{amsthm}
\usepackage{pifont}
\usepackage{xspace}
\usepackage{threeparttable}
\usepackage{color}
\usepackage{colortbl}
\usepackage{xurl}
\usepackage{hyperref}
\usepackage{subcaption}
\usepackage{paralist}
\usepackage{wrapfig}
\usepackage{multirow}
\usepackage{listings}
\usepackage{comment}
\usepackage{booktabs}
\usepackage{makecell}
\usepackage{setspace}
\usepackage{boxedminipage}
\usepackage{graphicx}
\usepackage{minibox}
\usepackage{pifont}
\usepackage{color}
\usepackage{xcolor}
\usepackage{subfiles}
\usepackage[utf8]{inputenc}
\usepackage[english]{babel}
\usepackage{ulem}
\usepackage{adjustbox}
\definecolor{blue}{RGB}{10, 10, 200}
\usepackage{amsmath}
\usepackage{titlesec}
\usepackage{enumitem}
\usepackage{amssymb}
\usepackage{diagbox}
\usepackage{balance}
\usepackage{multirow} 
\usepackage{booktabs}

\def\BibTeX{{\rm B\kern-.05em{\sc i\kern-.025em b}\kern-.08em
    T\kern-.1667em\lower.7ex\hbox{E}\kern-.125emX}}

\newcommand{\bheading}[1]{{\vspace{4pt}\noindent{\textbf{#1}}}}
\newcommand{\iheading}[1]{{\vspace{2pt}\noindent{\textit{#1}}}}
\newcolumntype{?}{!{\vrule width 1pt}}

\newcounter{note}[section]
\renewcommand{\thenote}{\thesection.\arabic{note}}

\ifx\submission\undefined
    \newcommand{\Sadoghi}[1]{{\color{cyan}{\textbf{Sadoghi:} #1}}}
    \newcommand{\niu}[1]{\refstepcounter{note}{\bf\textcolor{red}{$\ll$JN~\thenote: {\sf #1}$\gg$}}}
    \newcommand{\hz}[1]{\refstepcounter{note}{\bf\textcolor{purple}{$\ll$HZ~\thenote: {\sf #1}$\gg$}}}
    \newcommand{\shaokang}[1]{\refstepcounter{note}{\bf\textcolor{blue}{$\ll$SK~\thenote: {\sf #1}$\gg$}}}
    \newcommand{\Dakai}[1]
    {{\color{orange}{\textbf{Dakai:} #1}}}
    
    \newcommand{\yyu}[1]{\textcolor{black}{#1}}
    
     \newcommand{\y}[1]{\textcolor{black}{#1}}

\else
    \newcommand{\Sadoghi}[1]{}
    \newcommand{\niu}[1]{}
    \newcommand{\hz}[1]{}
    \newcommand{\shaokang}[1]{}
    \newcommand{\Dakai}[1]{}

\fi

\newcommand{\secref}[1]{\mbox{Sec.~\ref{#1}}\xspace}

\newcommand{\figref}[1]{\mbox{Fig.~\ref{#1}}}

\newcommand{\ignore}[1]{}

\newcommand{\ie}{\textit{i.e.}\xspace}
\newcommand{\eg}{\textit{e.g.}\xspace}

\newcommand{\sysname}{\textsc{Manta}\xspace}

\newcounter{packednmbr}
\newenvironment{packedenumerate}{
\begin{list}{\thepackednmbr.}{\usecounter{packednmbr}
\setlength{\itemsep}{0pt}
\addtolength{\labelwidth}{4pt}
\setlength{\leftmargin}{12pt}
\setlength{\listparindent}{\parindent}
\setlength{\parsep}{3pt}
\setlength{\topsep}{3pt}}}{\end{list}}

\newenvironment{packeditemize}{
\begin{list}{$\bullet$}{
\setlength{\labelwidth}{0pt}
\setlength{\itemsep}{2pt}
\setlength{\leftmargin}{\labelwidth}
\addtolength{\leftmargin}{\labelsep}
\setlength{\parindent}{0pt}
\setlength{\listparindent}{\parindent}
\setlength{\parsep}{1pt}
\setlength{\topsep}{1pt}}}{\end{list}}

\newtheorem{theorem}{Theorem}
\newtheorem{lemma}{Lemma}

\newtheorem{definition}{Definition}

\usepackage{tikz}
\usetikzlibrary{arrows.meta, shapes.geometric, positioning}

\newcommand{\Name}[1]{\textnormal{\textsc{#1}}}

\usepackage{algorithmic}
\usepackage{algorithm}

\newcommand{\Call}[2]{\Name{#1}(#2)}
\newenvironment{myprotocol}{
    \small
    \smallskip
    \algsetup{linenosize=\footnotesize}
    \begin{algorithmic}[1]
        
        \newcommand{\SPACE}{\item[]}
        
        \makeatletter
            \newcommand{\EVENT}[1]{\STATE \textbf{event} ##1 \textbf{do}\begin{ALC@g}}
            \newcommand{\ENDEVENT}{\end{ALC@g}}
        \makeatother
        
        \makeatletter
            \newcommand{\FUNCTION}[2]{\STATE \textbf{function} \Name{##1}(##2) \textbf{do}\begin{ALC@g}}
            \newcommand{\ENDFUNCTION}{\end{ALC@g}}
        \renewcommand{\COMMENT}[1]{\hfill $\triangleright$ ##1}
        \makeatother
    
}{
    \end{algorithmic}%
}

\newcounter{insightcounter}

\begin{document}

\title{\sysname: Unlocking DAG Flexibility in Asynchronous BFT}

\author{
{\rm Mohan Yu}\textsuperscript{1,$\dagger$},
{\rm Chenlin Wu}\textsuperscript{2,$\dagger$},
{\rm Jianyu Niu}\textsuperscript{2,$\ddagger$},
{\rm Cong Wang}\textsuperscript{2},
{\rm Chenyuan Wu}\textsuperscript{2}\\
\textsuperscript{1}Southern University of Science and Technology\\
\textsuperscript{2}City University of Hong Kong
} 
\maketitle
\begingroup
\renewcommand{\thefootnote}{\fnsymbol{footnote}}
\footnotetext[2]{These authors contributed equally.}
\footnotetext[3]{Part of the work was done when Jianyu Niu visited the Research Institute of Trustworthy Autonomous Systems of SUSTech.}
\endgroup

\begin{abstract}
DAG-based BFT protocols marked a shift in consensus design by replacing single-leader proposal paths with concurrent block proposals linked through a shared graph, where references encode consensus evidence. However, under heterogeneous wide-area networks (WANs), existing designs suffer from prolonged rounds and degraded performance due to straggler effects. The root cause is rigid round advancement: replicas move in synchronized rounds, waiting for enough references from the previous round before advancing.

We propose a flexible DAG growth framework for asynchronous BFT consensus. Our key observation is that while reference quorums are necessary for safety and liveness, they need not block every round advancement. Instead, the connectivity required for commitment can be accumulated across rounds. Within the framework, we instantiate \sysname, which combines \textit{flexible DAG growth} to advance by relaxing round synchronization and \textit{flexible coin} to trigger leader election and finality earlier. Experiments on heterogeneous WAN deployments show that \sysname reduces latency \yyu{by approximately 30\%–70\% in our evaluation.}

\end{abstract}

\section{Introduction} \label{sec:introduction}
Byzantine Fault-Tolerant (BFT) consensus underpins modern distributed systems, particularly geo-distributed services such as shared databases~\cite{Basil}, distributed storage~\cite{IPFS}, and blockchains~\cite{algorand,hanke2018dfinity}. It enables mutually untrusted nodes to agree on a common sequence of transactions despite Byzantine failures~\cite{castro1999practical,yin2019hotstuff}. In such geo-distributed settings, asynchronous BFT protocols~\cite{guo2020dumbo,li2025tockowl,duan2018beat,keidar2021all,miller2016honey} have emerged as an attractive choice, as they guarantee liveness without timing assumptions and remain robust under variable message delays, transient congestion, and WAN dynamics~\cite{miller2016honey}.

Recent DAG-based protocols have made asynchronous consensus practical at high throughput by alleviating the bottlenecks of earlier designs~\cite{guo2020dumbo,miller2016honey}. Examples such as DAG-Rider~\cite{keidar2021all} and Tusk~\cite{danezis2022narwhal} organize protocol execution as a Directed Acyclic Graph (DAG), where replicas continuously propose blocks, referred to as \textit{vertices}. 
Each vertex carries transaction payloads and typically references a quorum of at least $2f+1$ previously delivered vertices, forming directed edges that connect vertices. This minimum number of references, called the \textit{reference threshold}, governs round advancement and the connectivity of the growing DAG. Consensus is derived from the graph structure: once sufficient connectivity emerges, randomized leader election (\eg, common coin), combined with local DAG observations, derives a final total order and commits the causal history of selected leaders.

Yet the performance of existing DAG protocols still degrades sharply in heterogeneous WAN deployments, where cross-region links vary widely, transient slowdowns are common, and replicas often differ in bandwidth and responsiveness~\cite{enes2020state}. A few lagging replicas can slow system-wide progress, increasing commit latency and reducing throughput. Fig.~\ref{fig:dissectWAN} shows that under such heterogeneity, latency increases by around 40\%, throughput drops accordingly, and round duration grows by more than 25\%.

\begin{figure}[!t]
    \centering
    \begin{subfigure}{0.48\linewidth}
        \centering
        \includegraphics[width=\linewidth]{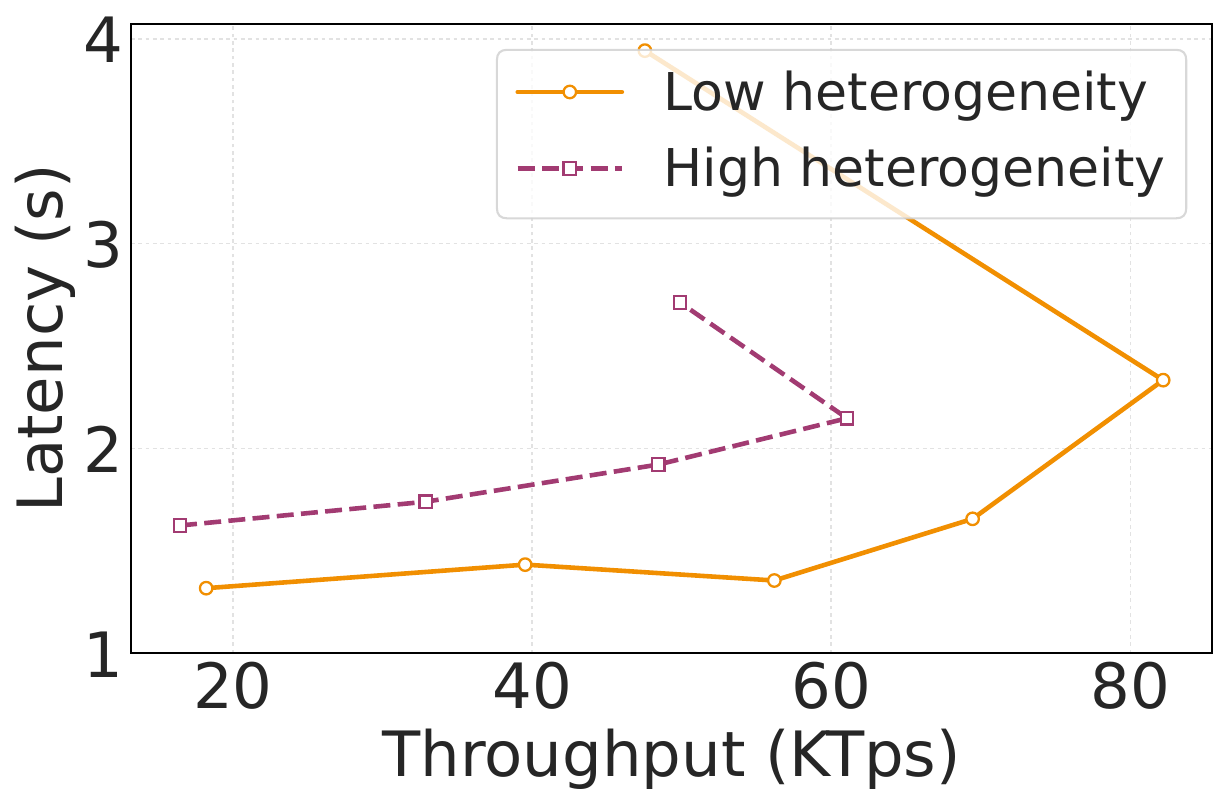}
        \caption{Latency and TPS.}
        \label{fig:intro_1}
    \end{subfigure}
    \hfill
    \begin{subfigure}{0.48\linewidth}
        \centering
        \includegraphics[width=\linewidth]{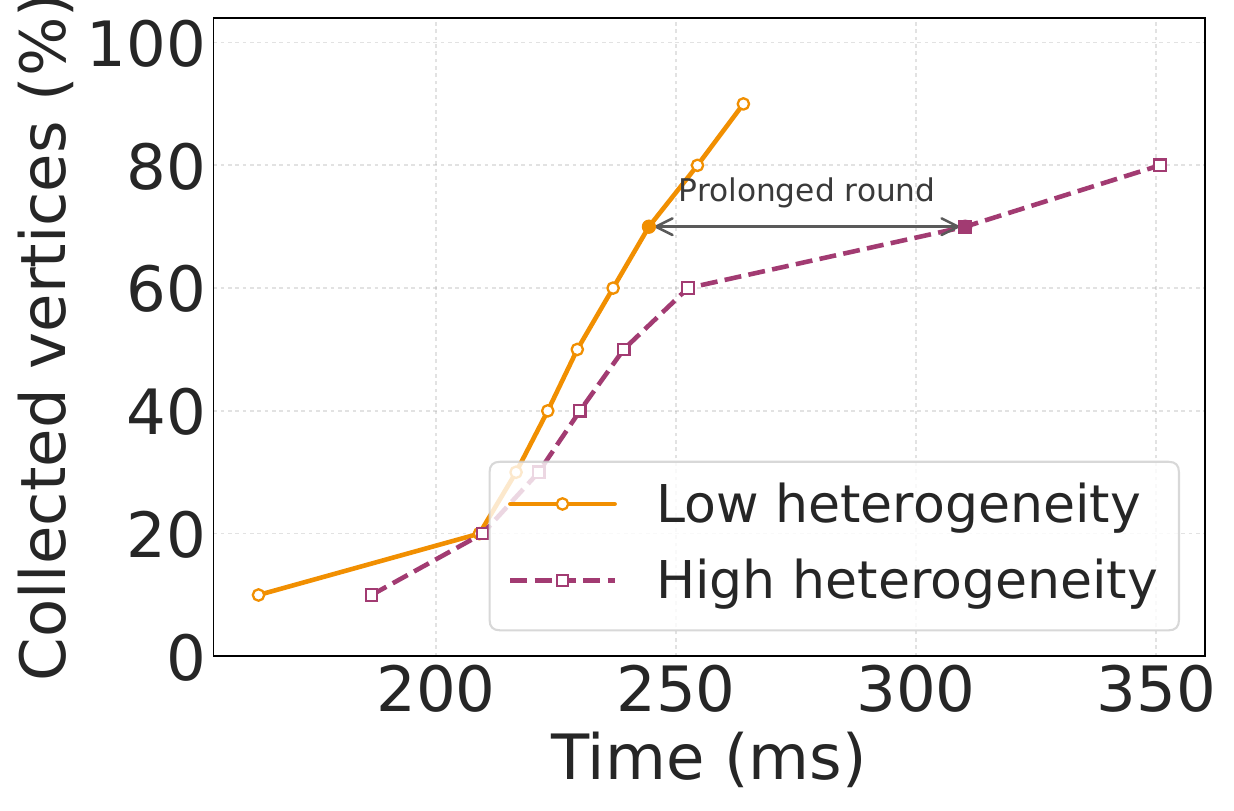}
        \caption{Prolonged round.}
        \label{fig:intro_2}
    \end{subfigure}
    \caption{Performance degradation in heterogeneous WANs.}
    \label{fig:dissectWAN}
\end{figure}

The root cause is a hidden synchronization barrier in round advancement: replicas must gather a quorum of references before moving to the next round~\cite{keidar2021all,jovanovic2025mahi,danezis2022narwhal,huang2025chitu}. This rule strengthens graph connectivity and increases the likelihood that elected leaders can commit even under adversarial schedules. Yet it also serializes progress: proposals are parallelized, but rounds still wait for stragglers.

Simply lowering the synchronization barrier is not enough. Reducing the reference threshold below $2f+1$ may shorten rounds, but in asynchronous adversarial settings, an adversary can selectively delay blocks, causing honest replicas to observe inconsistent structures and fail to guarantee liveness~\cite{keidar2021all}. 
More fundamentally, existing designs conflate two distinct goals: using the same per-round barrier both to establish connectivity and to regulate progress. This raises a key question: \textit{if strong connectivity is needed for safe commitment, must every round wait for it before making progress?}

Our answer is no. The synchronization barrier should serve as a condition for safety and liveness, not a blocking condition for round advancement. A high barrier preserves the connectivity required for stable commits, but enforcing it in every round is unnecessarily restrictive. In heterogeneous WANs, repeatedly waiting for slow previous-round vertices inflates latency and allows a few stragglers to dictate system progress.

Guided by this insight, we develop a flexible DAG growth framework for asynchronous BFT consensus. The framework makes explicit a fundamental tradeoff in DAG protocols: lowering the per-round barrier speeds up progress, yet can weaken the {connectivity} needed for commitment under adversarial schedules. It characterizes how reference threshold and connectivity accumulation affect commit latency, and derives liveness bounds that explain when aggressive designs with low threshold can fail. 
Building on this analysis, the framework introduces a coverage-based abstraction that decouples \textit{how much} connectivity is required from \textit{how} it is formed. 
Concretely, it provides boost edges that allow connectivity to accumulate across rounds, rather than being rebuilt in each round.
Finally, the framework defines a four-dimensional parameter space that abstracts the generic DAG growth process, enabling systematic analysis of existing protocols and the discovery of new ones.

Within the framework, we instantiate \sysname, a concrete asynchronous DAG-based protocol that converts already-formed connectivity into faster progress and earlier commits. \sysname is enabled by two complementary design elements: \textit{flexible growth} and \textit{flexible coin}. Flexible growth allows replicas to advance with lower reference thresholds while preserving worst-case liveness through cross-round coverage. Flexible coin then exploits locally emerged connectivity to trigger the common coin as soon as commitment becomes likely, enabling earlier leader election and finalization instead of waiting for fixed wave boundaries. If early commitment is unavailable, execution safely falls back to the normal path without affecting correctness or performance.

We built a prototype of \sysname and evaluated it against four state-of-the-art asynchronous counterparts: Tusk~\cite{danezis2022narwhal}, Chitu~\cite{huang2025chitu}, Mahi-Mahi~\cite{jovanovic2025mahi}, and DAG-Rider~\cite{keidar2021all}. We conducted experiments under diverse conditions on CloudLab~\cite{duplyakin2019design} and AWS platforms with up to 50 replicas. Results show that \sysname consistently outperforms other baselines \y{under fault-free heterogeneous deployments.}

\bheading{Contributions.} This paper makes four contributions:
\begin{packeditemize}
    \item We identify a fundamental limitation of existing asynchronous DAG-based protocols: rigid per-round synchronization barriers unnecessarily gate progress and inflate commit latency under heterogeneous WANs.

    \item We develop a unified framework for asynchronous DAG growth that formalizes the tradeoff between fast progress and sufficient connectivity for commitment, and explains the liveness limits of aggressive low-threshold designs. 

    \item We design \sysname, a novel asynchronous DAG-based BFT protocol that decouples connectivity from round advancement. Through flexible growth and flexible coin, \sysname achieves faster progress and earlier commitment while preserving worst-case safety and liveness.

    \item We implement \sysname and evaluate it against state-of-the-art protocols under realistic heterogeneous WANs. \sysname reduces latency by up to 70\%, improves throughput by up to 18\% at scale, and remains performant under faults.
\end{packeditemize}

\section{Background and Motivation}
\subsection{DAG-Based BFT Consensus} \label{sec:dag basic}
DAG-based BFT protocols proceed in \textit{rounds}, where consecutive rounds are grouped into \textit{waves}. In each round, every replica proposes a block, referred to as a \textit{vertex}. Each vertex references at least $2f+1$ vertices from the previous round via directed edges (often called \textit{strong edges}), resulting in a growing directed acyclic graph (DAG). 

Vertices are disseminated using reliable broadcast (RBC)~\cite{bracha1987asynchronous}, ensuring that all honest replicas eventually obtain a consistent view of the DAG (details in Appendix~\ref{app:building blocks}). As references accumulate, later-round vertices gain multi-hop reachability to earlier vertices, increasing the shared visibility needed for common core formation\yyu{~\cite{dolev2016some}}. If a subset of vertices from some round becomes reachable in the local DAGs of all honest replicas, these vertices form the \textit{common core} of that wave, meaning they are sufficiently \textit{observed} for subsequent commitment logic (formalized in \secref{subsec:problem_statement}).

At the end of each wave, replicas invoke a common coin~\cite{libert2014born,boneh2001short} to elect a leader vertex from the first round of the wave. If the elected leader belongs to the common core, it commits its causal history, \ie, all vertices reachable from it, according to a deterministic rule.

Existing DAG-based BFT protocols can be broadly categorized into two architectures: \textit{coupled} and \textit{decoupled}. In coupled designs, DAG vertices carry both transaction payloads and ordering information, so DAG growth directly determines transaction latency. In decoupled designs, data dissemination is separated from consensus, and vertices mainly carry metadata or certificates. Representative examples include DAG-Rider~\cite{keidar2021all} and Chitu~\cite{huang2025chitu} for coupled designs, and Tusk~\cite{danezis2022narwhal} and Autobahn~\cite{giridharan2024autobahn} for decoupled designs.
In this paper, we primarily focus on the coupled architecture by following prior asynchronous BFT protocols~\cite{keidar2021all, huang2025chitu}, while noting that our framework and evaluation also extend to decoupled settings.

\subsection{Dissecting DAG-Based BFT Performance} \label{subsec:observation}
\bheading{Heterogeneity in WAN deployments.}
Modern distributed systems like blockchains are typically deployed in wide-area networks, where replicas operate under highly heterogeneous conditions. Differences in networking (\eg, bandwidth), computation (\eg, CPU), and storage capabilities lead to uneven processing and dissemination speeds across replicas. In addition, client requests are often highly skewed~\cite{wei2017replication,liu2019distcache}, leading to substantial load imbalance across replicas~\cite{cooper2010benchmarking,li2020pegasus,weng2022mlaas,gai2023scaling}.

These effects are further amplified by heterogeneous network environments, where replicas are geographically distributed but unevenly placed, forming dense clusters alongside sparsely connected peripheral nodes~\cite{motepalli2023analyzing,motepalli2026gpos}. This results in highly variable network paths and message delays across replicas, as also observed in production blockchain systems such as Aptos~\cite{aptos_validators_explorer}, Solana~\cite{Solana_validators_datacenters}, and Sui~\cite{suivision_validators}.

\bheading{Rigid round progression constraint.}
Existing DAG-based BFT protocols couple round advancement with quorum collection: to enter round $r$, a replica must first gather at least $2f+1$ vertices from round $r\!-\!1$ (see \secref{sec:dag basic}). While this rule helps preserve connectivity and common-core formation, it also creates a per-round synchronization barrier. Under heterogeneous WANs, progress is dictated by the slowest quorum vertices rather than the fastest available ones. As a result, fast replicas remain idle with their resources underutilized, and prolonged rounds delay transaction admission and increase commit latency.

\yyu{Fig.~\ref{fig:dissectWAN} compares DAG-based protocols under low- and high-heterogeneity WAN deployments with the same average RTT ($80$\,ms) and total input workload (20K--120K tx/s).} In the former, all replica pairs have the same RTT and workloads are evenly distributed. In the latter, replicas are distributed across three continents and workloads are unevenly assigned, resulting in heterogeneous network delays and workload distribution. \yyu{Under heterogeneity, delayed cross-region arrivals prolong round duration by over 25\%, with commit latency around 40\% higher.}

\bheading{Key observations.} The above bottleneck stems from the coupling of round advancement with quorum collection. Breaking it reveals new opportunities for improving performance under heterogeneous conditions due to two key observations:

\begin{packeditemize}

\item \textbf{Connectivity should not be constrained by round boundaries.}
As discussed in Sec.~\ref{sec:dag basic}, references accumulated over time progressively increase reachability to earlier vertices. Existing protocols, however, enforce rigid quorum requirements within each round, forcing replicas to wait for all required vertices before proceeding. Under heterogeneous networks, many vertices often arrive early and already provide substantial partial connectivity before a wave fully completes. Since safety and liveness ultimately depend on whether enough vertices become sufficiently observed to form the common core, rather than on the exact timing of round completion, there is room to relax rigid round synchronization. Allowing connectivity to continue accumulating across rounds enables replicas to make progress earlier while preserving the guarantees required for safety and liveness.

\item \textbf{Commit rules should adapt to flexible connectivity.}
Relaxing round-based progression changes the structure of the DAG, as vertices are no longer aligned within strict rounds. As a result, commit rules that rely on consecutive-round connectivity become overly restrictive. To fully exploit flexible connectivity, commit conditions must be adapted to account for connectivity formed across rounds.
\end{packeditemize}

Guided by these two observations, \secref{sec:overview} develops a flexible DAG growth framework that theoretically characterizes and quantifies how rigid round-boundary constraints limit connectivity accumulation. \secref{sec:design} then presents \sysname, a concrete protocol instantiation that exploits the resulting DAG flexibility to accelerate commitment.

\vspace{-6mm}
\section{System Model} \label{sec:model} 
We consider a set of $n = 3f+1$ replicas, denoted by $\Pi=\{p_1,p_2,...,p_n\}$. Among them, up to $f$ replicas may be Byzantine, behaving arbitrarily under the control of a computationally bounded adversary, while the remainder are honest and faithfully follow the protocol. We assume replicas communicate over authenticated, reliable channels as in prior BFT works~\cite{miller2016honey}: messages sent by correct replicas are eventually delivered to all correct receivers without forgery. We assume a public key infrastructure (PKI) for identity verification.

We assume a fully asynchronous network model as prior works~\cite{keidar2021all,guo2020dumbo}, where message delays are unbounded, and delivery order is arbitrary. In particular, message delays may vary significantly across replicas, capturing heterogeneous network conditions in practice. We also assume standard reliable broadcast (RBC) primitives~\cite{bracha1987asynchronous} that ensure agreement, integrity, and validity properties~\cite{guerraoui2019scalable}, and a shared random coin~\cite{libert2014born,boneh2001short} (\eg, instantiated via threshold signatures) to support randomized coordination under asynchrony.

We focus on state machine replication (SMR)~\cite{schneider1990implementing}, where replicas agree on a total order of transactions. The protocol should satisfy two properties: (i) \textit{Safety}: no two correct replicas commit different transactions at the same position in the total order; and (ii) \textit{Liveness}: every transaction issued by a correct client is eventually committed by honest replicas.

\section{Flexible DAG Growth Framework} \label{sec:overview}
To unlock flexible DAG growth, we present a unified analytical framework that formalizes the key design trade-offs: while lowering the per-round synchronization barrier can shorten round duration, it also weakens connectivity and thus affects the chance of successful commitment. 
To this end, we first model commit probability and latency that are jointly determined by protocol design parameters and network conditions (\secref{subsec:problem_statement}), and then analyze worst-case protocol behaviors (\secref{subsec: ref_adv}), showing that DAG structures that perform well in practice may fail to guarantee liveness under adversarial settings. These insights not only reveal the limitations of existing DAG protocols, but also guide the design of more flexible and robust DAG frameworks (\secref{subsec: framework}). 

\subsection{Modeling DAG Performance} \label{subsec:problem_statement}
To analyze the commit latency of DAG-based BFT protocols, we begin by examining how vertices propagate and become visible across replicas. 

\bheading{Common core ($S$).} In DAG-based protocols, the leader vertex can be committed only if it is sufficiently well disseminated so that all honest replicas \textit{observe} it. Specifically, consider a replica $p$ proposing a vertex $v$ in round $r$. We say that $p$ observes a vertex $u$ in round $r'$ (where $r' < r$) if there exists a path from $v$ to $u$.
Furthermore, we have a key abstraction for reasoning about commit decisions. 

\begin{definition}
[Common core~\cite{dolev2016some}] Let $V_i$ denote the set of the first-round vertices that are observed by replica $p_i$ after a wave. A set $S$ is called a common core if $S \subseteq V_i$ for every honest replica $p_i$. 
\end{definition}

\begin{figure}[!t]
    \centering
    \includegraphics[width=0.9\linewidth]{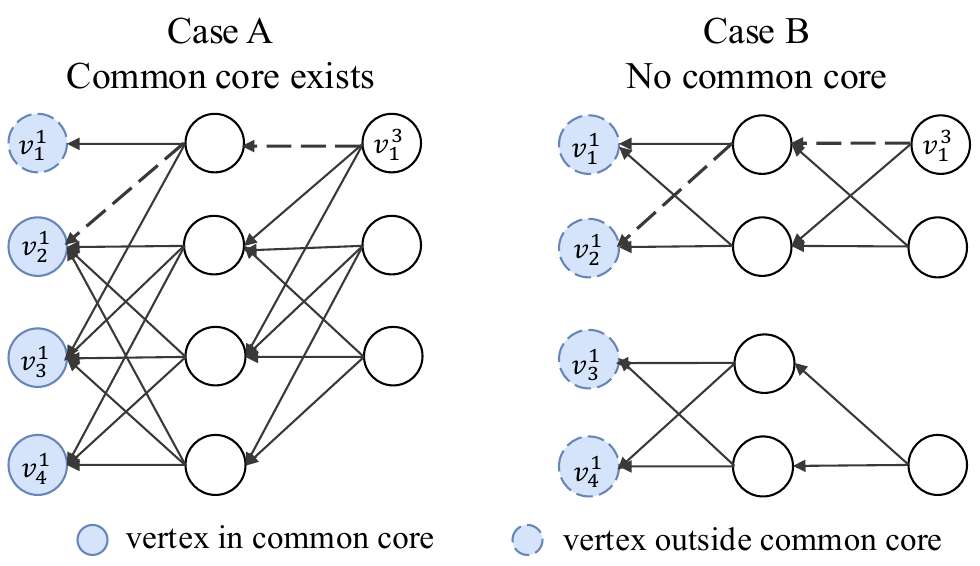}
    \caption{\yyu{Common core example for $k= 3$ rounds in a wave with $4$ replicas ($f=1$). A first-round vertex is observed if there exists a path from a third-round vertex to it. Hence, $v_1^3$ observes $v_2^1$ via the dashed path. \textbf{Left}: With $ref=2f+1$, vertices $v_2^1$, $v_3^1$, and $v_4^1$ are observed by all honest replicas and thus form the common core, whereas $v_1^1$ is observed by only one replica and lies outside the common core. \textbf{Right}: With $ref=f+1$, no first-round vertex is observed by all honest replicas, resulting in an empty common core.}}
    \label{fig:commoncore}
    \vspace{-4mm}
\end{figure}

\iheading{Example.} Consider a wave where each replica observes a slightly different subset of vertices due to heterogeneous delays. A common core is a \yyu{ set of vertices observed by all honest replicas}. \yyu{Fig.~\ref{fig:commoncore} shows how different reference thresholds can change this shared visible set.} \yyu{For a vertex in the common core, its commit status can therefore be evaluated consistently across replicas.} In contrast, a vertex outside the common core is not guaranteed to be observed by some honest replica at the decision point.

We next identify two key protocol parameters that govern how connectivity forms. The reference threshold $ref$ determines how many vertices must be collected before a replica advances to the next round, while the wave length $k$ determines how long connectivity is allowed to accumulate before a commit decision is made. \yyu{We use $\Phi(ref,k)$ to characterize the effective visible set size by the end of round $k$ under reference threshold $ref$; together with Byzantine behavior, it characterizes the size of $S$ (see Appendix~\ref{app: hetero modeling}).}

\bheading{Commit probability ($P$).} \yyu{With the leader selected uniformly at random from the first round, any leader in the common core is guaranteed to be committable. Thus, the commit probability is given by}
\[
P \ge \frac{\Phi(ref, k)}{n}.
\]

\bheading{Wave duration ($W$).} A replica advances after collecting $ref$ references from the previous round. Let $\{d_1, \dots, d_n\}$ denote the propagation delays of vertices from different replicas. The time required to complete a round is determined by the arrival of the $ref$-th fastest vertex, \ie, the $ref$-th order statistic $D = d_{(ref)}$. Over a wave of $k$ rounds, the total time spent in accumulating connectivity is denoted as $W(ref,k)$. 

\bheading{Latency ($L$).} Since a successful commit occurs with probability $P$, the latency estimate until commitment is 
\[
L \approx \frac{W(ref,k)}{P}.
\]

\subsubsection{Optimistic Analysis under Heterogeneity} Using the above formulation, we next analyze the performance of DAG-based BFT protocols under benign conditions without adversarial interference. 
We model $d$ as a log-normal random variable to capture different levels of heterogeneity and then derive the theoretical latency. Due to space constraints, detailed proof is provided in Appendix~\ref{app: hetero modeling}.

\begin{figure}[!t]
    \centering
    \begin{subfigure}{0.48\linewidth}
        \centering
        \includegraphics[width=\linewidth]{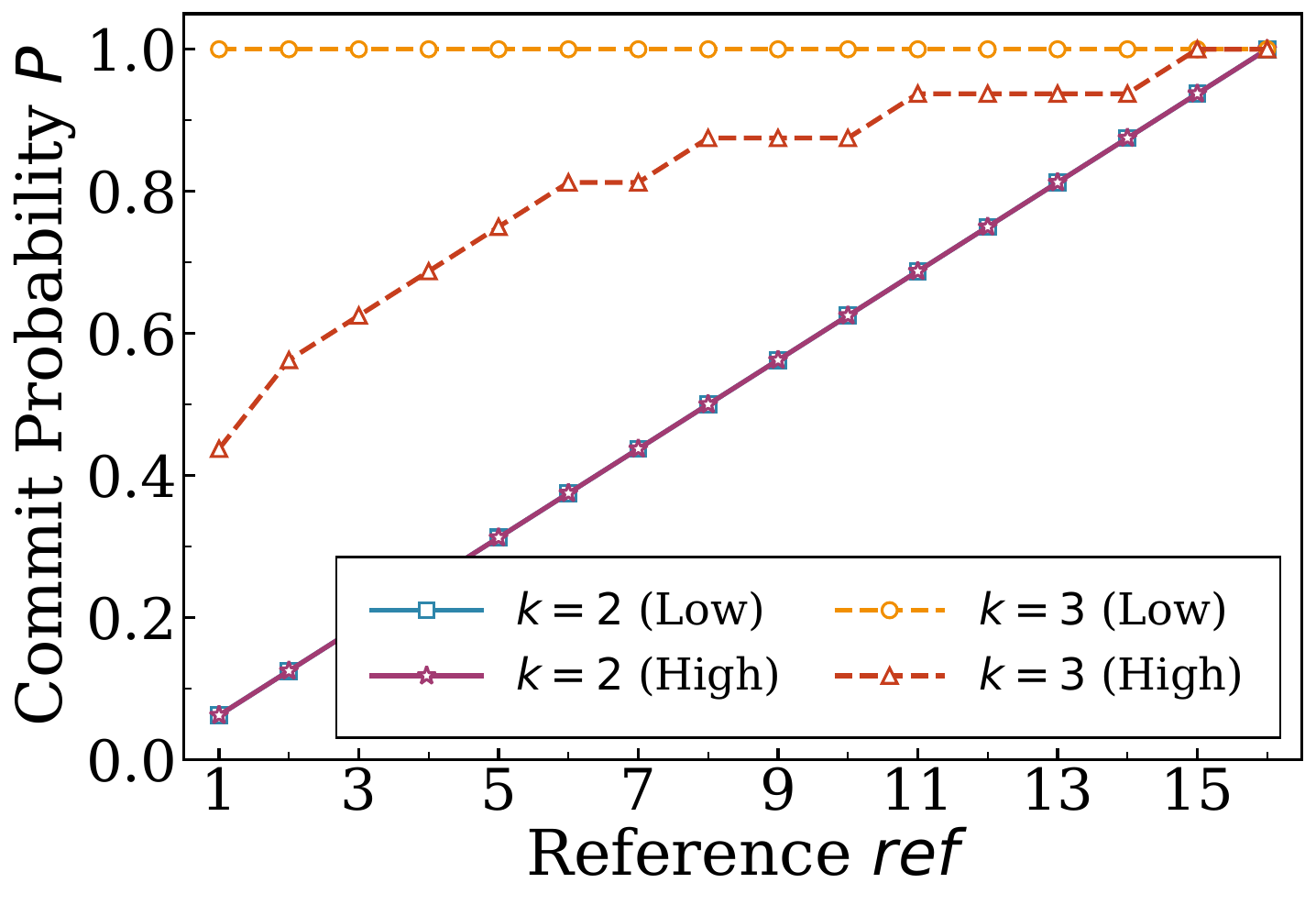}
        \caption{Commit probability.}
        \label{fig:commit P}
    \end{subfigure}
    \begin{subfigure}{0.48\linewidth}
        \centering
        \includegraphics[width=\linewidth]{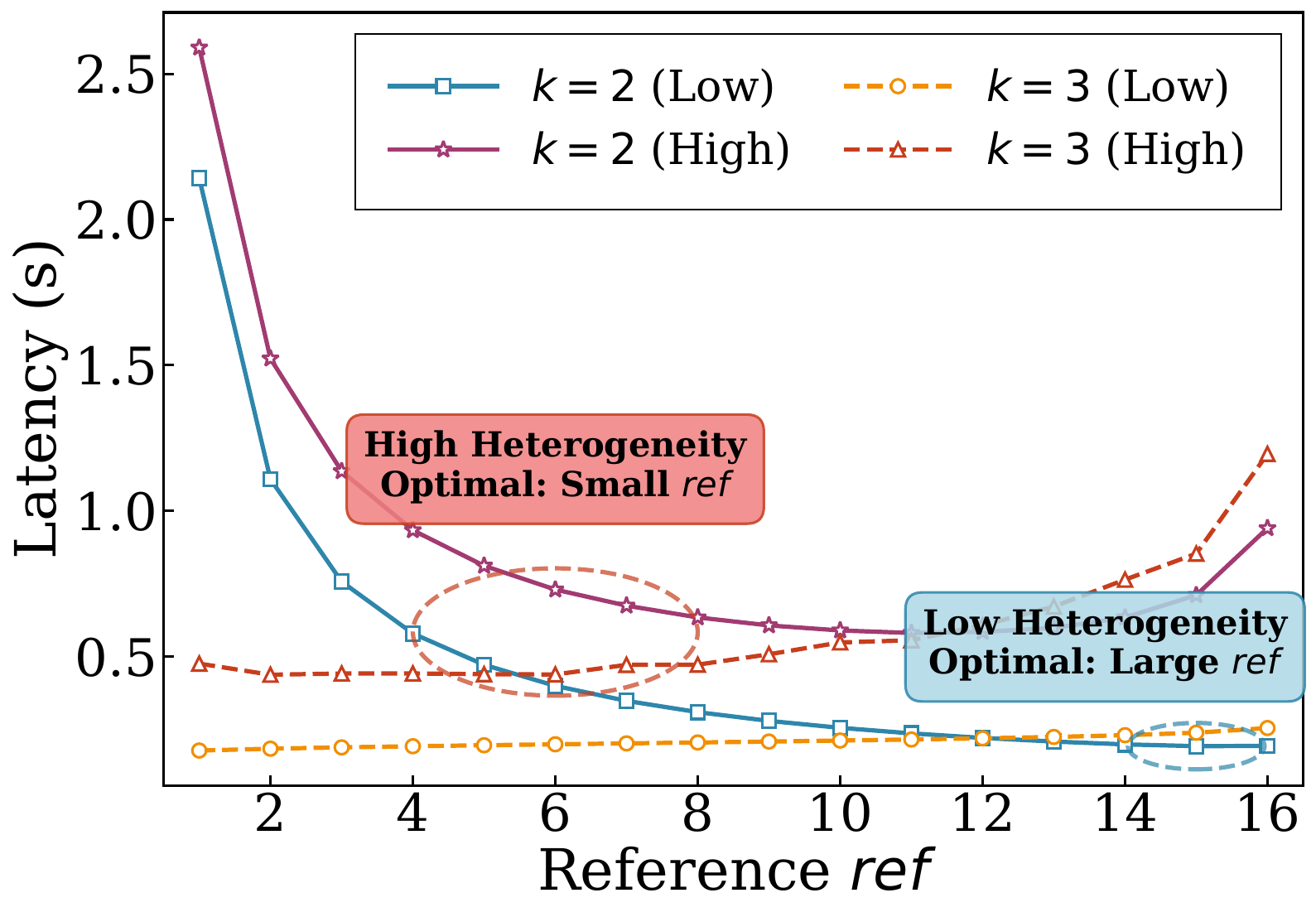}
        \caption{Latency.}
        \label{fig:latency_vs_heterogeneity}
    \end{subfigure}
    \caption{DAG performance under different $ref$ and $k$.}
    \label{fig:analysisResults}
    \vspace{-4mm}
\end{figure}

Fig.~\ref{fig:analysisResults} shows the theoretical commit probability and latency under different choices of $ref$ and $k$. \yyu{We instantiate the distribution of $d$ using representative network-delay distributions with mean/std of 50/5 ms and 100/50 ms for the favorable and unfavorable settings, respectively.} The results reveal a clear trade-off between round latency and connectivity accumulation. When heterogeneity is high, long-tail delays from slow replicas make large $ref$ prohibitively expensive. Although increasing $ref$ improves commit probability, this gain is insufficient to offset the increased round duration. In contrast, a smaller $ref$ allows replicas to advance earlier, and the saved time can be used to accumulate connectivity across additional rounds. As a result, connectivity can emerge through temporal accumulation rather than strict per-round synchronization, leading to lower latency. 

In summary, we observe distinct optimal configurations under different levels of heterogeneity:
\begin{packeditemize}
    \item \textbf{Low heterogeneity.} A larger $ref$ with a smaller $k$ minimizes latency by maximizing per-round connectivity.
    \item \textbf{High heterogeneity.} A smaller $ref$ combined with a larger $k$ achieves better performance by leveraging cross-round accumulation.
\end{packeditemize}

\subsection{Liveness Constraints} \label{subsec: ref_adv}
The above optimistic analysis may motivate adopting aggressive configurations such as $ref = f+1$ to minimize protocol latency across heterogeneous networks. However, such configurations cannot be directly applied as explained below. \yyu{The following analysis examines the liveness issue behind these configurations and motivates the solid vertex mechanism introduced in Sec.~\ref{subsec: framework}.}

\bheading{Why optimistic configurations fail.}
{The limitation arises because reducing $ref$ \yyu{may leave insufficient connectivity} under adversarial timing.
{The right side of Fig.~\ref{fig:commoncore} (Case B) shows a case in asynchronous networks where an adversary can delay message delivery such that no first-round vertex is observed by all honest replicas within a fixed wave length.
In this case, the honest replicas share no common visible first-round vertex, so the common core is empty. As a result, no vertex is guaranteed to be committed, and progress depends on DAG topology and message timing, preventing convergence and threatening liveness.} 

The liveness issue is exacerbated in practical two-round designs. As illustrated in Case A of Fig.~\ref{fig:core_k2}, when $ref=f+1$, an adversary can prevent every first-round vertex from receiving \yyu{at least} $f+1$ supporting edges, which is the weakest possible commit condition. Hence, the wave contains no committable vertex, and liveness fails. \yyu{In contrast, in Case B with $ref=2f+1$, some first-round vertices satisfy the commit rule, preserving liveness.} 

\bheading{Worst-case liveness guarantees.}
To formalize these limitations, we analyze the worst-case behavior of DAG-based BFT protocols. Let $R_1, \ldots, R_k$ denote the sets of vertices generated in each round of a wave. {For a fixed decision round $R_k$, the key question is which vertices in $R_1$ remain committable under any admissible adversarial DAG topology, since these vertices determine the protocol's worst-case liveness guarantee. \yyu{Since short-wave DAG protocols, \eg, those that check commitment after two rounds, may not form a common core even when some first-round vertices remain committable~\cite{shoup2025blue}, we characterize liveness directly through these committable vertices}. \yyu{We therefore generalize the notion of common core and} let $G_k \subseteq R_1$ denote the set of first-round vertices that are guaranteed to satisfy the commit rule (with threshold $f+1$) in the worst case, and let $U_k = R_1 \setminus G_k$ denote the remaining vertices, whose committability is topology-dependent and thus not guaranteed in the worst case. }

\begin{figure}[!t]
    \centering
\includegraphics[width=0.9\linewidth]{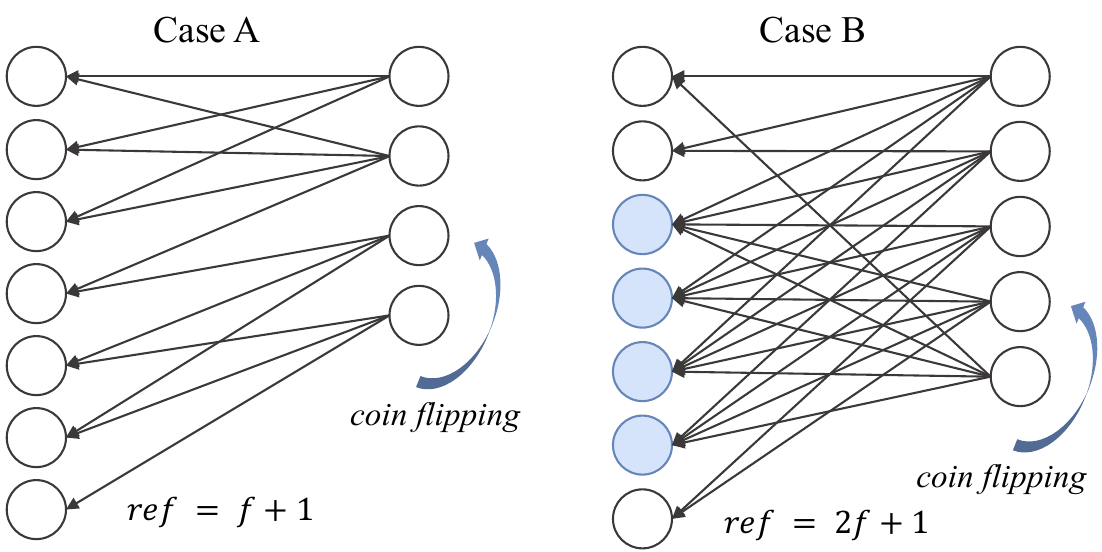}
    \caption{A worst-case two-round example with $f=2$.}
    \label{fig:core_k2}
    \vspace{-2mm}
\end{figure}

\begin{theorem}[Worst-case common-core bound] \label{the: common core}
Consider a DAG-based BFT protocol with $n=3f+1$ replicas, wave length $k$, and reference threshold $ref$. Under adversarial scheduling, the size of $G_k$ is lower-bounded as summarized in Table~\ref{tab:worst_case_s_size}. In particular:

(i) when $ref = f+1$, we have $|G_k| = 0$ for $f>1$; and

(ii) when $ref \ge n-f$, the $G_k$ remains non-empty, with the guaranteed size increasing with $k$.
\end{theorem}

\iheading{Proof Sketch.} {The proof partitions the first-round vertices into the guaranteed set $G_k$ and the uncertain set $U_k$. \yyu{It then lower-bounds $|G_k|$ by counting the support to $R_1$ that is unavoidable under adversarial scheduling.} It first identifies when a vertex in $R_1$ must enter $G_k$: \yyu{for \(k=2\), we first maximize the references that can be absorbed while keeping all \(R_1\) vertices outside \(G_2\). The remaining references must force some vertices into \(G_2\), yielding the lower bound on \(|G_2|\). For \(k=3\), a vertex enters \(G_3\) once it has at least \(n-ref+1\) supporters in \(R_2\), ensuring that any choice of \(ref\) parents in \(R_3\) includes at least one such supporter. We then maximize the references that can be absorbed while keeping vertices outside \(G_3\); any remaining references must force some vertices into \(G_3\), yielding the stated bounds.} The full proof is deferred to Appendix~\ref{subsec: worstcase}.} 

Theorem \ref{the: common core} shows that naively reducing $ref$ to improve latency fundamentally compromises liveness. {For example, if a protocol aims to use a two-round wave, \ie, to start commitment checking at the end of round 2, the reference threshold must remain above the corresponding worst-case bound to ensure that at least one first-round vertex satisfies the commit condition. More generally, if a protocol aims for stronger worst-case liveness, so that more honest replicas can make progress concurrently, it needs a larger common core, which in turn requires a larger $ref$ and $k$ according to Table~\ref{tab:worst_case_s_size}.}

\begin{table}[!t]
\centering
\caption{Common Core Size.}
\label{tab:worst_case_s_size}
\renewcommand{\arraystretch}{1.15}
\setlength{\tabcolsep}{7pt}
\scalebox{0.9}{
\begin{tabular}{@{}ccccc@{}}
\toprule
\diagbox[trim=l,width=3.5em]{$k$}{$ref$} & $\ge n-f$ & $\ge 1.8f+1$ & $\ge 1.5f+1$ & $\ge f+1$ \\
\midrule
2 & $f+1$  & 1        & --        & -- \\
3 & $1.5f$ & $f+1$    & 1         & -- \\
4 & $2f+1$ & $1.8f+1$ & $1.5f+1$  & -- \\
\bottomrule
\end{tabular}
}
\end{table}

\subsection{From Rigidity to Flexibility} \label{subsec: framework} 
The worst-case analysis above suggests that liveness fundamentally depends on whether sufficient connectivity is formed to preserve a non-empty $G_k$, rather than on how such connectivity is physically constructed. In other words, different reference patterns could induce the same effective connectivity, and thus yield identical commit probability and liveness guarantees. For example, \figref{fig:Redundant} shows two different reference patterns \yyu{(the dashed lines in Cases A and B) with the same guaranteed set $G_k=\{v_1^1,v_2^1,v_3^1\}$, whose vertices satisfy the commit rule, and thus the same commit probability $P$.}

Inspired by this observation, \yyu{we seek to preserve the connectivity required for a non-empty $G_k$ without enforcing a fixed reference pattern.} \yyu{To this end}, we introduce a more general notion of coverage for DAG, denoted as $cov$, which represents the number of vertices in a target round that are reachable from a given future vertex, regardless of the specific reference paths. Under this abstraction, the protocol no longer needs to enforce a monolithic reference rule (\eg, every vertex references $2f+1$ in the previous round); instead, connectivity can be flexibly constructed as long as the required coverage is achieved. Specifically, consider a vertex $v$ in round $r_{cur}$, and its goal is to reference at least $cov$ vertices in round $r_{cov}$ (where $r_{cov} \le r_{cur}-1$). To achieve this, $v$ establishes two types of edges:

\bheading{Progress edges.} They reference vertices in the previous round ($r_{cur}-1$). These edges allow $v$ to inherit references to $r_{cov}$ from its parents.

\bheading{Boost edges.} They reference vertices in any historical round $r' < r_{cur} -1 $ within the same wave, {increasing coverage over the target round when the inherited view from parent vertices is insufficient (\eg, directly referencing vertices in $r_{cov}$ that are not covered by the parents).} 

Boost edges facilitate cross-round integration to improve the commit probability. Therefore, the advancement conditions can be safely adjusted through the following mechanism.
\begin{definition}[Solid Vertex]
For a target round $r_{cov}$, let $I(v, r_{cov})$ denote the set of target-round vertices visible to $v$ through its parents, and let $B(v, r_{cov})$ denote the additional target-round vertices reached via boost edges. A vertex $v$ is solid with respect to round $r_{cov}$ if

$$
|I(v, r_{cov}) \cup B(v, r_{cov})| \ge cov.
$$

\end{definition}

\begin{figure}[!t]
    \centering
\includegraphics[width=1\linewidth]{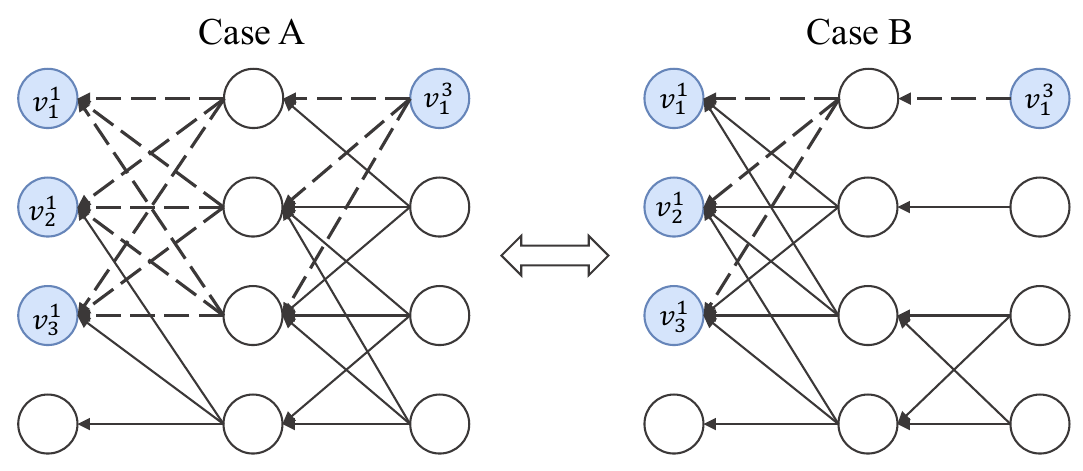}
    \caption{Redundant connectivity in DAG structures. Different reference patterns yield the same commit probability $P$ when the guaranteed set $G_k$ is identical. Dashed lines in Cases A and B illustrate example reference patterns to the same round-1 vertices.} 
    \label{fig:Redundant}
    \vspace{-2mm}
\end{figure}

Rather than constraining the DAG to specific reference patterns, we could guarantee liveness by verifying whether vertices are solid, shifting the focus to coverage ($cov$): the total number of vertices in a target round $r_{cov}$ reached by a vertex in the current round.
{To illustrate, consider $k=2$ and $ref=2f+1$ in Table~\ref{tab:worst_case_s_size}, where at least $f+1$ vertices in the first round are guaranteed to satisfy the commit rule. {After introducing $cov$, preserving this guarantee only requires preserving the same effective support to the target round. Therefore, setting $cov=ref=2f+1$ is sufficient:} even if the support is realized through different topologies, the resulting worst-case liveness bound remains unchanged. \figref{fig:multi-ref} shows that this support can indeed be achieved through multiple topology patterns:
 
\begin{packeditemize}
\item \textbf{Case A:} Round 2 vertices establish $ref = 2f+1$ progress edges to round 1. Consequently, round 3 vertices only require a minimum $ref = 1$ to round 2 to satisfy the $cov$ requirement for round 1.
\item \textbf{Case B:} Round 2 vertices maintain a minimal $ref = 1$ to round 1. Round 3 vertices must utilize progress and boost edges to collectively reference at least $cov$ vertices from round 1.
\item \textbf{Case C:} Round 2 vertices provide partial coverage with $ref = f+1$. Round 3 vertices then supplement this by referencing additional vertices via progress or boost edges until the $cov$ requirement for round 1 is fully met.
\end{packeditemize}

{Although the three cases utilize different edge compositions, they all render the vertices in round 3 solid. Therefore, they all guarantee the same worst-case liveness outcome: by round 3, at least $f+1$ vertices in round 1 satisfy the commit rule. In this way, the protocol can preserve liveness when the direct reference threshold $ref$ is reduced in some rounds for faster round progression, effectively decoupling the required reference pattern from the physical DAG construction.}

\begin{figure}[!t]
    \centering
\includegraphics[width=0.9\linewidth]{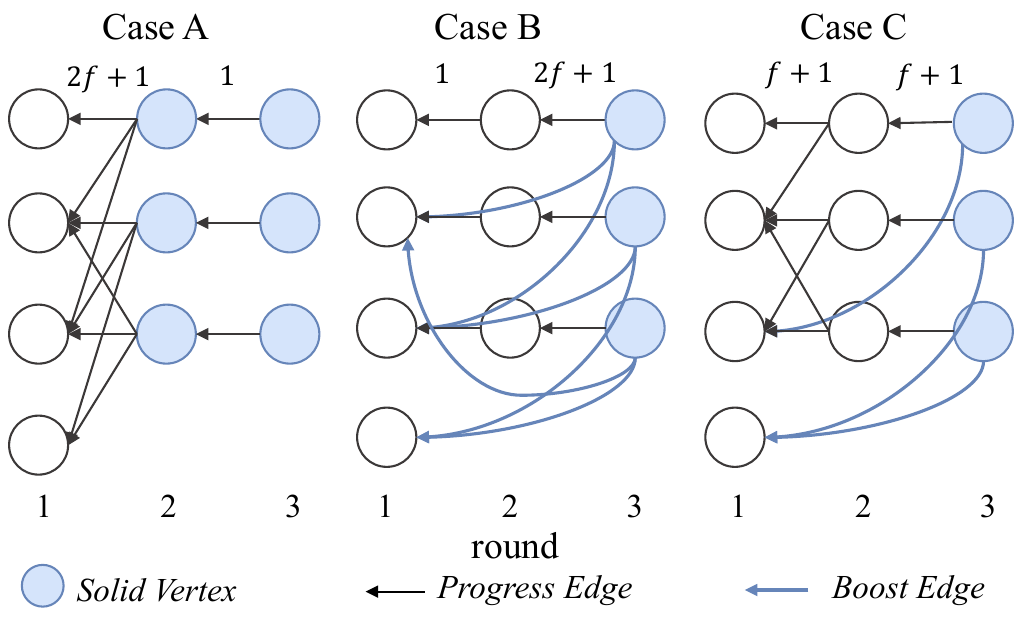}
    \caption{Multiple edge configurations achieving equivalent $cov$ thresholds. For a fixed wave length $k$, different reference patterns can satisfy the same coverage requirement.}
    \label{fig:multi-ref}
\end{figure}

\section{\sysname Design} \label{sec:design}
We present \sysname, a DAG-based BFT protocol instantiated from the previous framework with flexible DAG growth to reduce latency while preserving liveness. 
We first establish a unified parameter space for DAG design and use it to instantiate \sysname (\secref{subsec:overview}), then present the detailed protocol specification (\secref{subsec: Protocol}), provide the correctness analysis (\secref{subsec: Correctness}), and map existing DAG protocols to the design space behind \sysname (\secref{subsec:mapping}).

\subsection{Instantiation of \sysname} \label{subsec:overview}
We instantiate \sysname based on the analysis in \secref{sec:overview}. To better support flexible DAG growth, we first define \textit{solid step} that comprises multiple rounds.

\begin{definition}
    [Solid Step] A sequence of rounds $[r, r+\sigma]$ is termed a \textit{Solid Step} of span $\sigma$ if every honest replica proposes at least one vertex within these rounds that is a \textit{Solid Vertex} regarding round $r$.
\end{definition}

We then define a parameter space $\Omega = (ref, cov, \sigma, \kappa)$ that abstracts the DAG growth process. Specifically, replicas propose a vertex in each round, and a wave consists of $\kappa$ consecutive solid steps; each solid step spans $\sigma$ rounds to achieve the target commit coverage $cov$. \yyu{Subsequent proposals preserve the same coverage by causally extending the solid vertex.}

The choice of $cov$ for each solid step and the number of solid steps $\kappa$ per wave determine the resulting common core size in the worst case, suggesting that different configurations of $\Omega$ correspond to different levels of liveness guarantees. 
Specifically, to optimize for latency under heterogeneous settings, we instantiate \sysname by setting $cov = 2f+1$, $\sigma=2$, and $\kappa=2$, as shown in Fig.~\ref{fig:overview}. Thus, a wave spans five rounds and contains two solid steps. Every honest replica must propose at least one solid vertex per step. 

{The parameters $cov$ and $\kappa$ are chosen to ensure worst-case liveness, as characterized in Table~\ref{tab:worst_case_s_size}. We set $\sigma=2$ to support flexible DAG growth and choose $ref=f+1$, which is identified by the analysis in~\secref{subsec:problem_statement} as the favorable design point for high heterogeneity environments. Moreover, it leaves room for the flexible coin design without making every round trivially solid.}

\begin{figure}[!t]
    \centering
\includegraphics[width=1\linewidth]{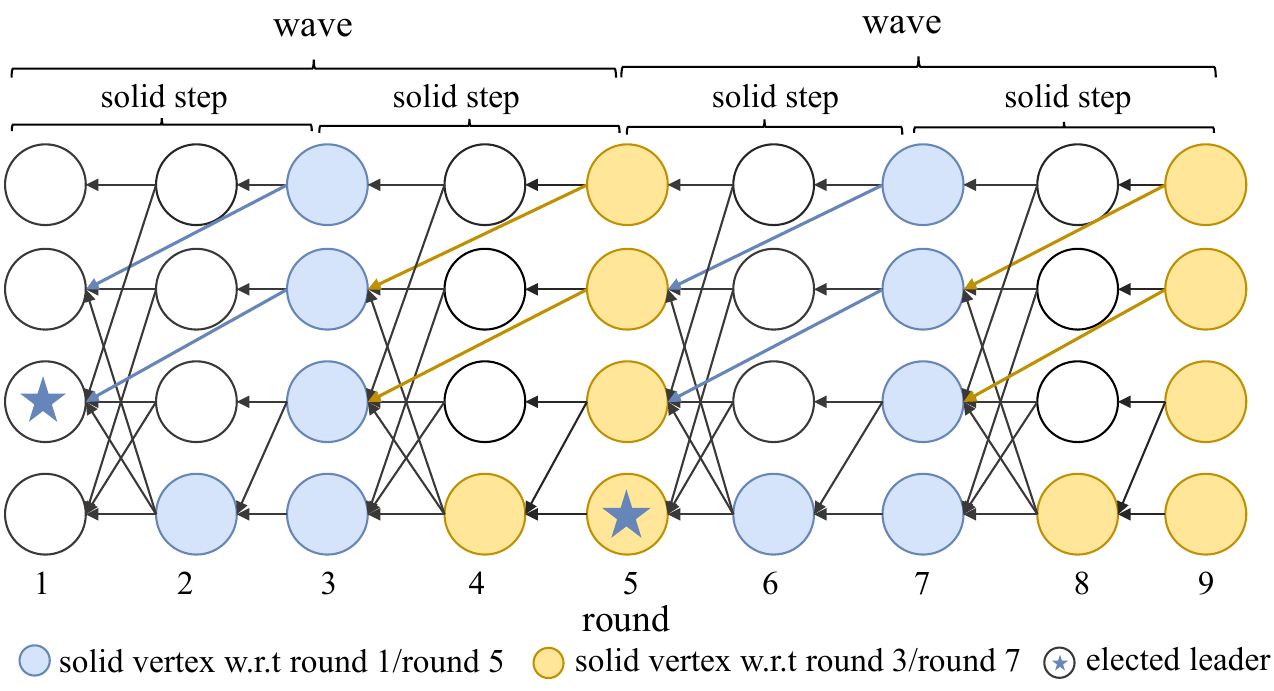}
    \caption{The wave structure in \sysname. A wave consists of $\kappa=2$ solid steps, each spanning $\sigma=2$ rounds. }
    \label{fig:overview}
\end{figure}

\subsection{Protocol in Detail} \label{subsec: Protocol}
Algorithm~\ref{alg:sysname-core} shows the core utilities of \sysname, which leverages the flexible parameter space $\Omega$ to optimize performance in heterogeneous networks. 

\noindent\ding{182}\textit{Vertex construction.} Each vertex carries a batch of transactions together with references to earlier vertices, which are interpreted as progress or boost edges. To advance to the next round, a replica must propose a vertex that either references at least $ref$ vertices from the immediately preceding round or constitutes a \textit{solid vertex} for the current step (Lines~\ref{line: proposeV}-\ref{line: endproposeV}). 

\noindent\ding{183} \textit{DAG construction.} 
As in prior DAG-based BFT protocols, vertices are disseminated via reliable broadcast. Replicas accept at most one vertex from each sender in a given round, vote for a valid proposal, and certify a vertex after collecting $n-f$ votes. Once certified, the vertex may be referenced by subsequent proposals (details in Appendix~\ref{app:dag structure}).

\noindent\ding{184} \textit{Ordering and commitment.} Each replica locally interprets its DAG view to establish a consistent total order. {To achieve consensus, replicas use a common coin at the end of the wave (\ie, in round 5 in \figref{fig:overview}) to elect a leader from the first round of that wave. The commitment decision is then made by checking votes on the leader in round 3: if the elected leader in round 1 is observed by at least $f+1$ vertices in round 3, it is committed together with its causal history (Lines~\ref{line:tryRegular}-\ref{line:endtryRegular}).}
Specifically, there are two commitment cases:

\begin{packeditemize}
    \item \textit{Direct commit.}
    {Let $v_l$ be the leader of wave $w$ determined by the recovered coin. We use $r_{obs}$ to denote the observation round, \ie, the round whose vertices are used to check whether $v_l$ has sufficient support for commitment. An honest replica commits $v_l$ if it is revealed by the recovered coin and is observed by enough vertices in round $r_{obs}$. In the regular path, we set $r_{obs}=3$, which aligns with the default end-of-wave schedule (Lines~\ref{line:tryRegular}-\ref{line:endtryRegular}).}

    \item  \textit{Recursive commit.} Leader $v_l$ is recursively committed if there exists a leader vertex $v_l'$ from a subsequent wave $w' > w$ that is directly committed and has a path to $v_l$.
\end{packeditemize}

\begin{algorithm}[!t]
    \caption{\sysname\ Core Utilities}
    \label{alg:sysname-core}
    \begin{myprotocol}
        \STATE \textbf{Local variables:} $\mathit{DAG}[\ast] \gets \emptyset$, $\mathit{round} \gets 1$
        \STATE \hspace{1.65cm} $\mathit{committed}[\ast] \gets \emptyset$, $\mathit{committedRound} \gets 0$
        \SPACE

        \FUNCTION{proposeVertex}{}      \label{line: proposeV}
            \STATE $P \gets$ \Call{getParents}{$\mathit{round}-1$}
            \IF{$|P| \ge ref$ \textbf{or} \Call{isSolidCandidate}{$\mathit{round}, P$}}
                \STATE $v \gets$ \Call{makeVertex}{$\mathit{round}, P$}
                \STATE \Call{broadcast}{$v$}
                \STATE $\mathit{round} \gets \mathit{round} + 1$
            \ENDIF
        \ENDFUNCTION
        \SPACE \label{line: endproposeV}

        \FUNCTION{nextCommittedVertices}{}
            \STATE $r \gets \mathit{committedRound} + 1$
            \IF{$r+1 \le \mathit{round}$}
                \STATE \Call{tryFastCoin}{$r$}
            \ENDIF
            \IF{$r > \mathit{committedRound}$ \textbf{and} $r+2 \le \mathit{round}$}
                \STATE \Call{tryFlexibleCoin}{$r$}
            \ENDIF
            \IF{$r > \mathit{committedRound}$ \textbf{and} $r+2 \le \mathit{round}$}
                \STATE \Call{tryRegular}{$r$}
            \ENDIF
            \IF{$\mathit{committedRound} \ge r$}
                \STATE \textbf{return} $\mathit{committed}[r]$
            \ENDIF
            \STATE \textbf{return} $\emptyset$
        \ENDFUNCTION
        \SPACE

        \FUNCTION{tryFastCoin}{$r$} \label{line:tryfastcoin}
            \IF{\Call{candidates}{$r, r+1$} $\ge \tau_{\mathit{cand}}$} 
                \STATE \Call{tryCommit}{$r, r+1$}
            \ENDIF 
        \ENDFUNCTION
        \SPACE \label{line:endfastcoin}

        \FUNCTION{tryFlexibleCoin}{$r$} \label{line:tryflexiblecoin}
            \IF{\Call{candidates}{$r, r+2$} $\ge \tau_{\mathit{cand}}$} 
                \STATE \Call{tryCommit}{$r, r+2$}
            \ENDIF
        \ENDFUNCTION
        \SPACE \label{line:endtryflexiblecoin}

        \FUNCTION{tryRegular}{$r$} \label{line:tryRegular}
            \STATE \Call{tryCommit}{$r, r+2$}
        \ENDFUNCTION
        \SPACE \label{line:endtryRegular}

        \FUNCTION{tryCommit}{$r, r_{\mathit{obs}}$}
            \STATE $v_l \gets$ \Call{getLeaderVertex}{$r$} \COMMENT \textit{using the common coin}
            \IF{$v_l \neq \bot$ \textbf{and} $\left| \{ v \in \mathit{DAG}[r_{\mathit{obs}}] : \Call{path}{v_l, v} \} \right| \ge f+1$}
                \STATE \Call{commitCausalHistory}{$v_l, r$}
            \ENDIF
        \ENDFUNCTION
        \SPACE
        
        \FUNCTION{candidates}{$r, r_{\mathit{obs}}$}
            \STATE \textbf{return} $\left| \left\{ u \in \mathit{DAG}[r] : \left| \left\{ v \in \mathit{DAG}[r_{\mathit{obs}}] : \Call{path}{u, v} \right\} \right| \ge f+1 \right\} \right|$
        \ENDFUNCTION
        \SPACE

    \STATE \textbf{upon receiving certified vertex} $v$ \textbf{do}
    \STATE \hspace{0.5cm} $\mathit{DAG}.\Call{addVertex}{v}$
    \end{myprotocol}
\end{algorithm}

\bheading{Flexible coin optimization.}
Since the connectivity required for commitment may emerge at different speeds under different network conditions, revealing common coin shares only at fixed round boundaries can introduce unnecessary delay. To address this, \sysname designs a flexible coin mechanism: beyond the default end-of-wave schedule, a replica may piggyback its threshold-signature share onto its next vertex once its local DAG view indicates sufficient connectivity for commitment. If enough shares are collected early, replicas recover the common coin sooner and elect the first-round leader earlier; otherwise, the shares remain available for later recovery along the regular path. The early-release trigger is controlled by a parameter $\tau_{\mathit{cand}}$, which affects only how aggressively early recovery is attempted and does not affect correctness. The resulting commitment check may either use the regular observation round $r_{obs}=3$ (Lines~\ref{line:tryflexiblecoin}-\ref{line:endtryflexiblecoin}) or, in a more aggressive fast-coin case, further advance the observation round to $r_{obs}=2$ (Lines~\ref{line:tryfastcoin}-\ref{line:endfastcoin}); to preserve safety in this case, \sysname requires an extended solid step from round 2 to round 5 (details in Appendix~\ref{app:correctproof}).

\iheading{Example.} Consider a 4-round wave structure where the leader is elected in the fourth round, as shown in \figref{fig:fastcoin}. In Cases A and B, sufficient connectivity already emerges by the second round, so releasing coin shares at that point allows replicas to recover the common coin and commit the leader earlier, whereas waiting until the default wave boundary would only add latency without improving commit probability.

\begin{figure}[!t]
    \centering
    \includegraphics[width=0.9\linewidth]{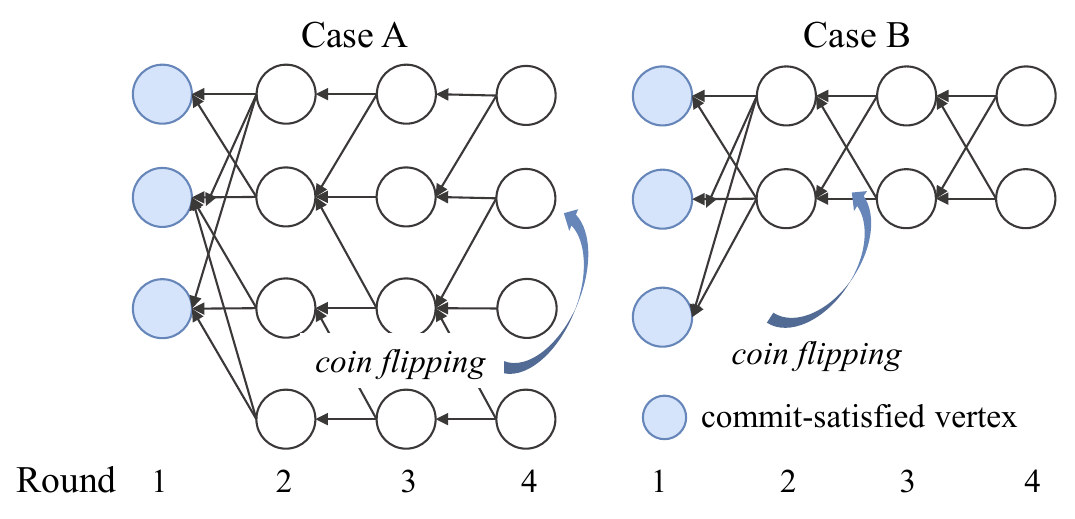}
    \caption{Early leader finalization via flexible coin.}
    \label{fig:fastcoin}
    \vspace{-3mm}
\end{figure}

\subsection{Correctness Analysis} \label{subsec: Correctness}
We sketch the safety and liveness arguments of \sysname here and defer full proofs to Appendix~\ref{app:correctproof}.

\bheading{Safety.}
Along the regular path, a leader vertex $v_l$ of wave $w$ (at round 1) can be committed if it is observed by at least $f+1$ vertices in round 3. By the \textit{solid step} structure, every vertex in round 5 must reference a quorum of at least $2f+1$ round-3 vertices. Since $n=3f+1$, quorum intersection guarantees that these $2f+1$ referenced vertices and the $f+1$ round-3 observers share at least one common vertex. Hence, every round-5 vertex (the first round of wave $w+1$) must observe $v_l$ and include it in its causal history. Therefore, every leader committed in a later wave contains all earlier committed leaders in its causal history. Since replicas linearize each leader's causal history using a deterministic ordering rule, all honest replicas derive the same total order of committed vertices.

\bheading{Liveness.} \sysname uses randomness to elect a leader in round 5. A round-1 leader $v_l$ is committed once it is observed by at least $f+1$ round-3 vertices. Because each round-3 vertex must observe a quorum of $2f+1$ round-1 vertices, quorum intersection ensures that, even under worst-case scheduling, there exists a non-empty set of at least $f+1$ round-1 candidates satisfying the commitment condition. Thus, in each wave, the probability that the random coin selects a committable leader is at least $(f+1)/(3f+1)$. Since this probability is strictly positive and waves repeat indefinitely, a leader is eventually committed with probability 1.

\subsection{Encoding Existing DAG-BFT Protocols} \label{subsec:mapping}

\sysname provides a unified abstraction for DAG-based BFT constructions. Using the four-dimensional parameter space $\Omega=\{ref, cov, \sigma, \kappa\}$, we map five representative DAG-based protocols to discrete design points within a common framework. Although these systems differ in consensus logic and network assumptions, their DAG growth patterns can be interpreted as specific instantiations of our abstraction.

\begin{packeditemize}
    \item \textbf{Tusk~\cite{danezis2022narwhal} and DAG-Rider~\cite{keidar2021all}:}
    Both protocols use $\sigma=1$ with dense references $ref=cov=2f+1$, so each round forms a solid step. In our framework, Tusk corresponds to $\kappa=2$, while DAG-Rider uses $\kappa=3$ to improve liveness at the cost of one additional round.

    \item \textbf{Chitu~\cite{huang2025chitu}:}
    Chitu introduces a leaderless fast path to avoid unnecessary fallback to more expensive consensus procedures. Structurally, it uses $\sigma=1$ and $ref=cov=2f+1$, making each round a solid step. Within our framework, Chitu dynamically switches between two wave configurations: a fast path with $\kappa=1$ when observations are sufficiently aligned, and a normal path with $\kappa=3$ when views diverge, and randomized leader selection is needed.

    \item \textbf{Mahi-Mahi~\cite{jovanovic2025mahi}:}
    Mahi-Mahi operates over an uncertified DAG with $\sigma=1$ and $ref=cov=2f+1$. Depending on the network model and desired commit probability, it can be represented with either $\kappa=3$ or $\kappa=4$ solid steps per wave.

    \item \textbf{Autobahn~\cite{giridharan2024autobahn}:}
    Autobahn adopts a decoupled architecture in which data lanes and consensus progress independently under partial synchrony. We focus on its DAG-style data layer, modeling $\sigma$ as a dynamic variable to capture non-uniform lane advancement. $cov$ denotes the lane coverage required to form a consensus cut, typically $2f+1$.

\end{packeditemize}

\section{Evaluation} \label{sec:evaluation}
We build a prototype of \sysname and compare it against four state-of-the-art  counterparts: Tusk~\cite{danezis2022narwhal}, Chitu~\cite{huang2025chitu}, Mahi-Mahi~\cite{jovanovic2025mahi}, and DAG-Rider~\cite{keidar2021all}, under various conditions. 
Our evaluation aims to answer three questions:
\begin{packeditemize}

    \item \textbf{Q1:} How do protocol parameters in the unified design space trade off with each other, and do the results align with our theoretical insights? (Sec.~\ref{subsec:theory-align})

    \item \textbf{Q2:} How does \sysname compare with existing DAG-based protocols under heterogeneous networks, and how resilient is it to node faults? (Sec.~\ref{subsec:comparison})

    \item \textbf{Q3:} How effective is each component of \sysname under heterogeneous networks, and what resource utilization does it incur? (Sec.~\ref{subsec:breakdown})
\end{packeditemize}

\begin{figure*}[!t]
    \centering
       \resizebox{1\textwidth}{!}{%
    \begin{minipage}{\textwidth}
    \centering
    
    \begin{subfigure}[b]{0.25\linewidth}
        \centering
        \includegraphics[height=3.5cm, keepaspectratio]{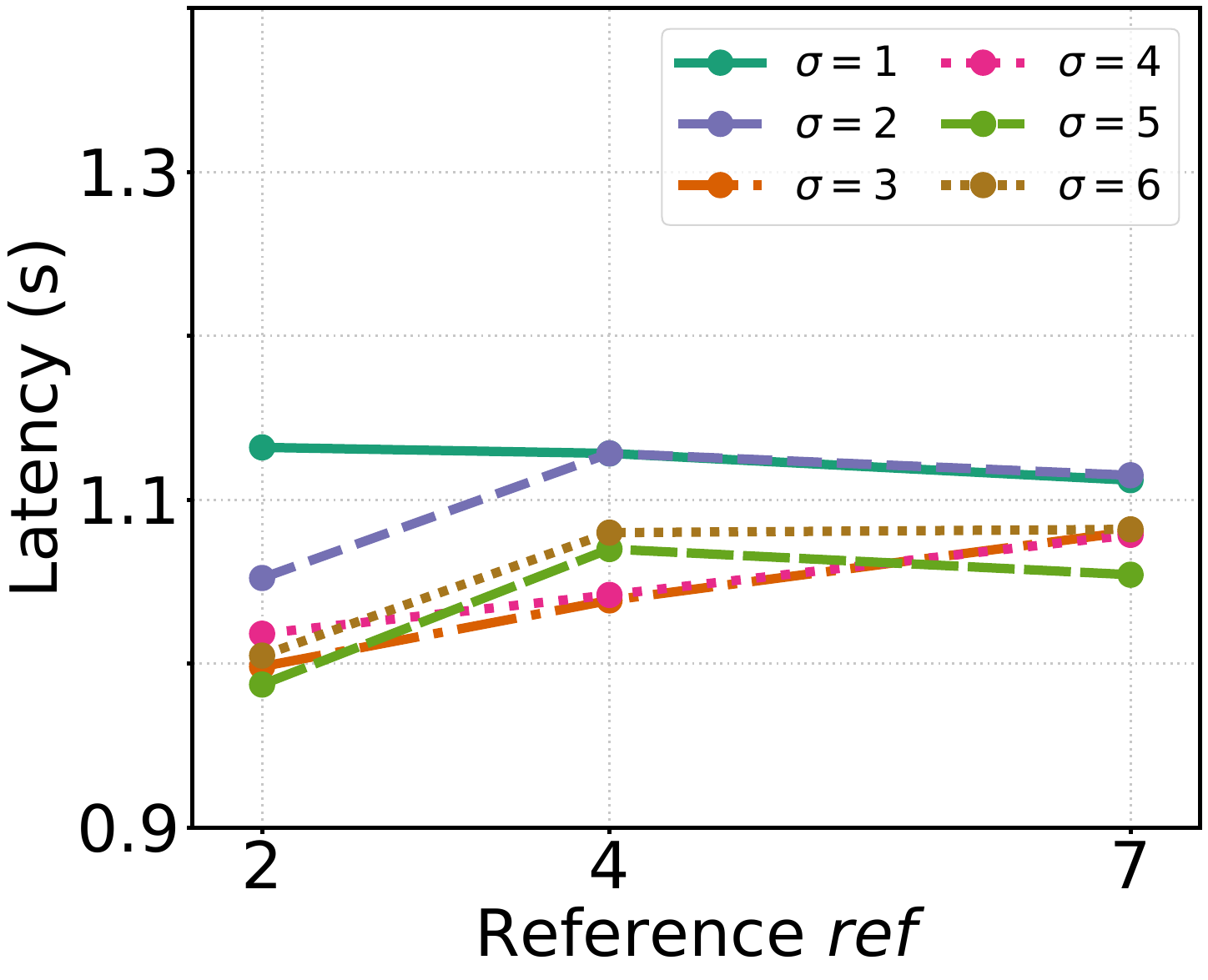}
        \caption{Effect of $\sigma$ and $\mathit{ref}$ under N1.}
        \label{fig:exp2-a}
    \end{subfigure}%
    \begin{subfigure}[b]{0.25\linewidth}
        \centering
        \includegraphics[height=3.5cm, keepaspectratio]{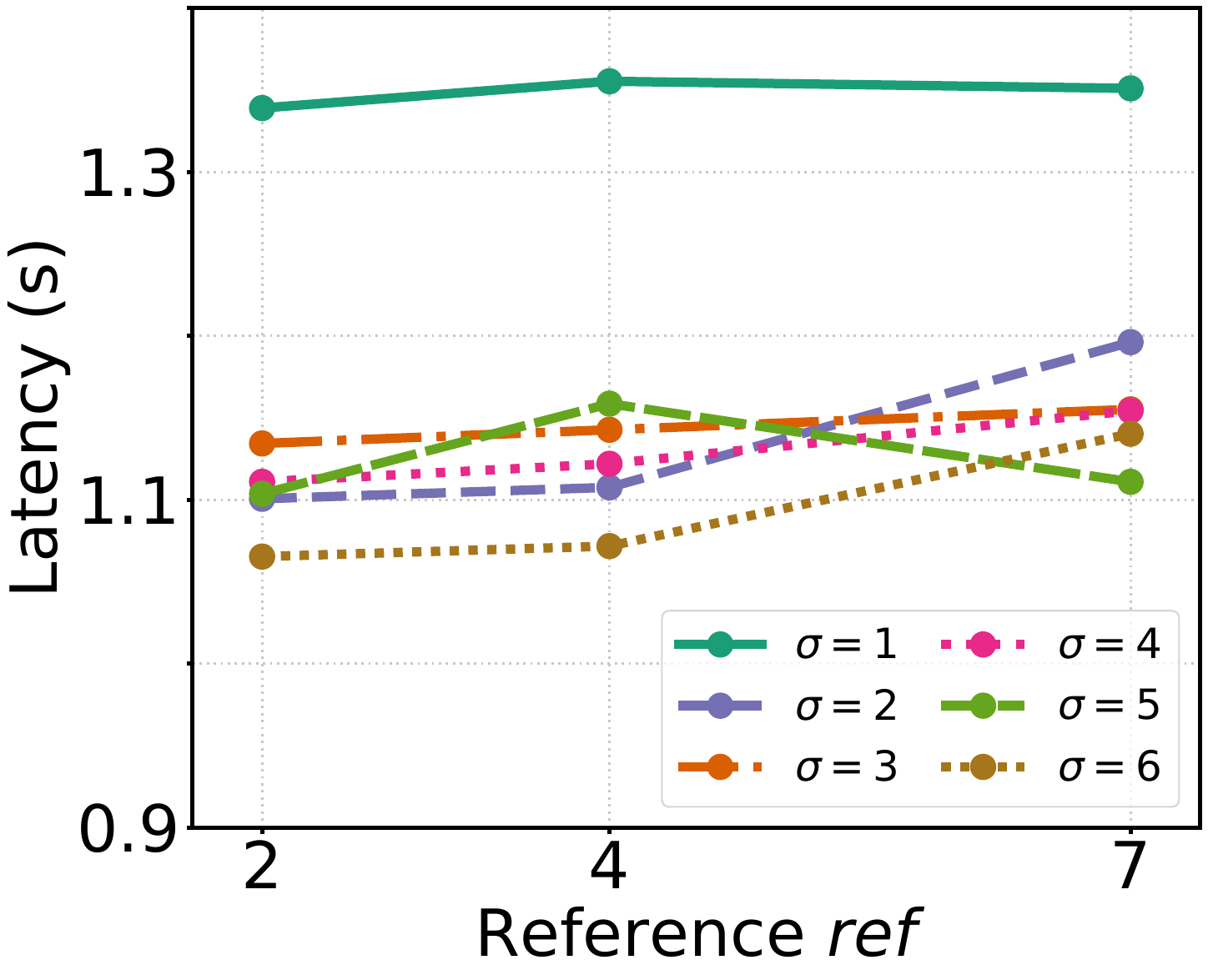}
        \caption{Effect of $\sigma$ and $\mathit{ref}$ under N2.}
        \label{fig:exp2-b}
    \end{subfigure}%
    \begin{subfigure}[b]{0.50\linewidth}
        \centering
        \raisebox{0.1em}{
            \includegraphics[height=3.5cm, keepaspectratio]{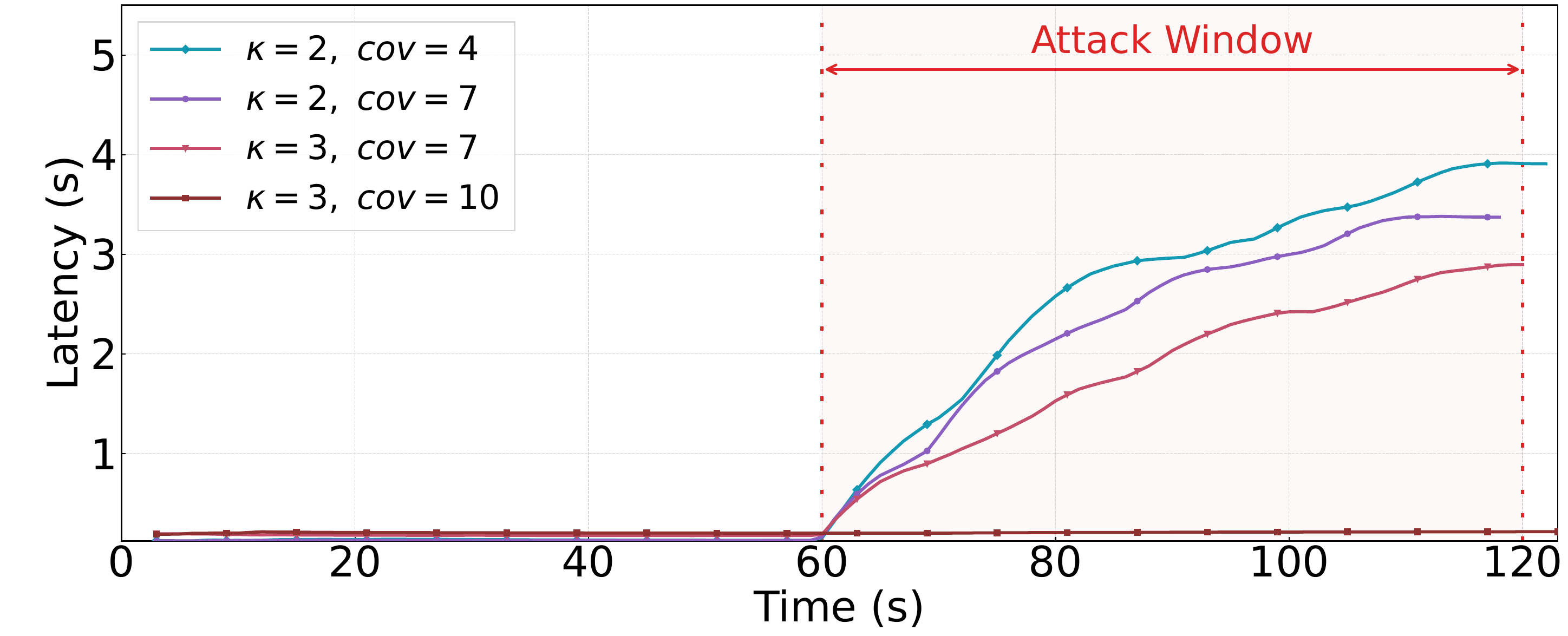}
        }
        \caption{$\kappa,\mathit{cov}$ with adversaries.}
        \label{fig:exp3-kappa-Byzantine}
    \end{subfigure}

    \par\medskip

    \begin{subfigure}[b]{0.25\linewidth}
        \centering
        \includegraphics[height=3.5cm, keepaspectratio]{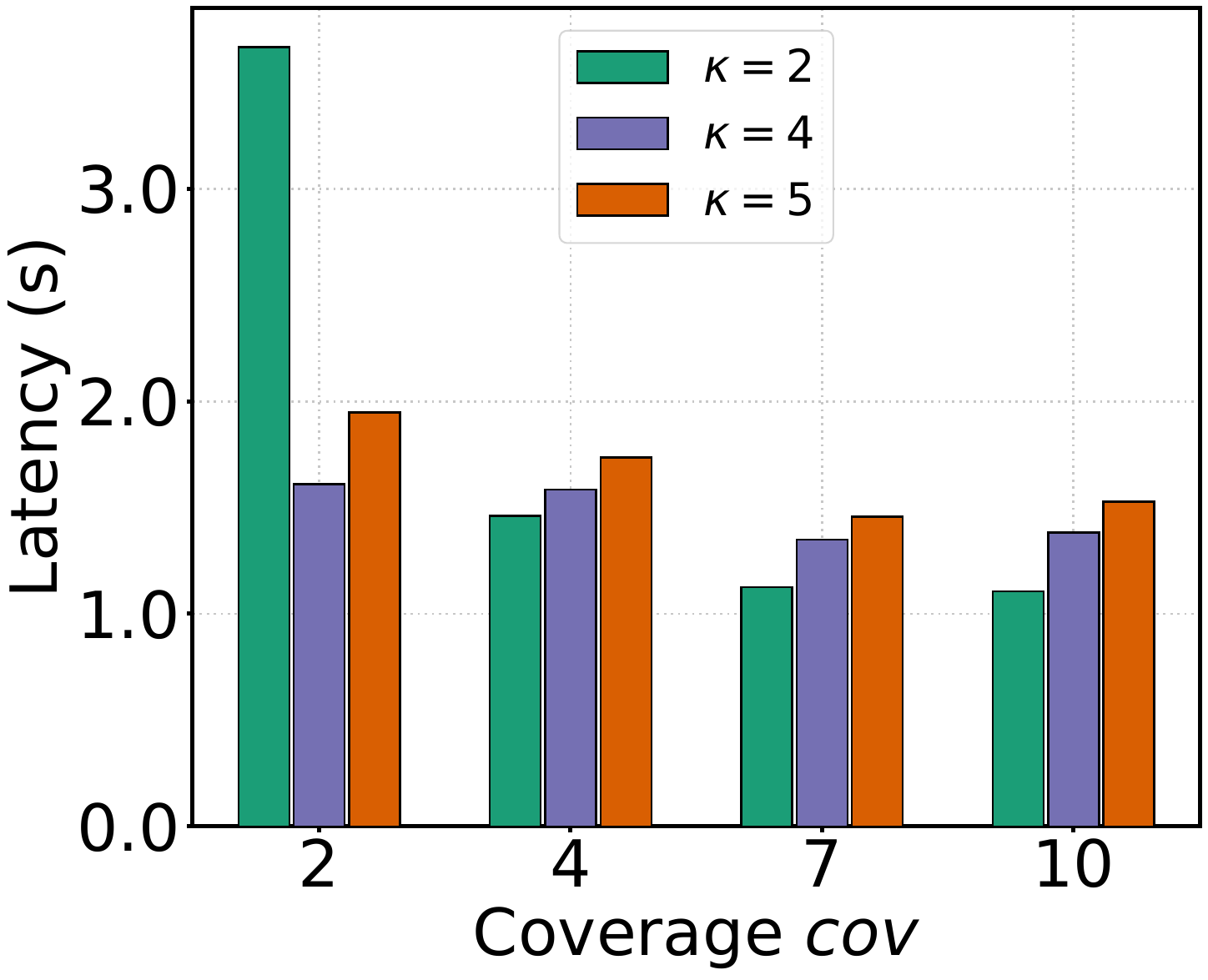}
        \caption{$\kappa$, $\mathit{cov}$ under N1.}
        \label{fig:exp2-d}
    \end{subfigure}\hfill
    \begin{subfigure}[b]{0.25\linewidth}
        \centering
        \includegraphics[height=3.5cm, keepaspectratio]{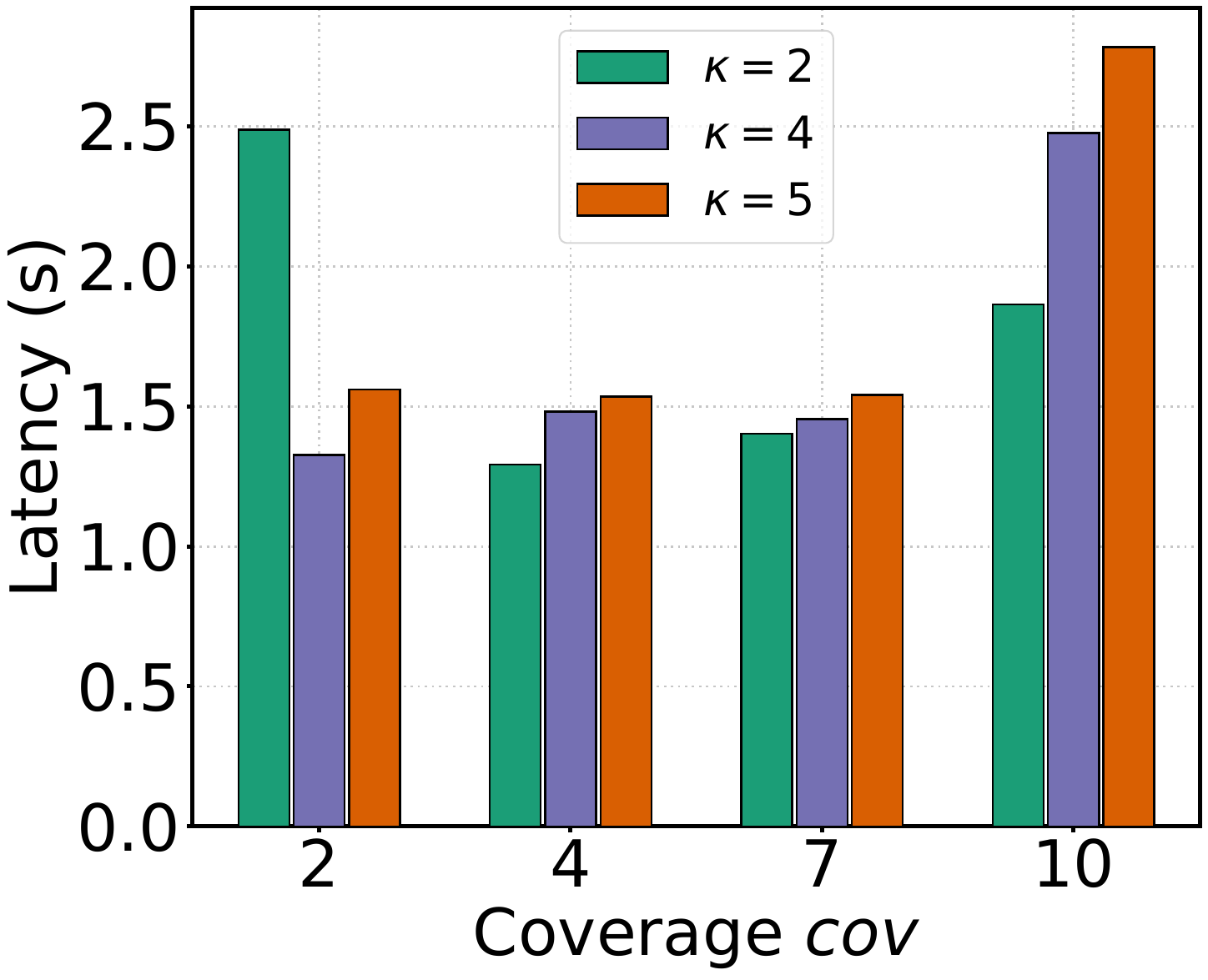}
        \caption{$\kappa$, $\mathit{cov}$ under N2.}
        \label{fig:exp2-e}
    \end{subfigure}\hfill
    \begin{subfigure}[b]{0.25\linewidth}
        \centering
        \includegraphics[height=3.5cm, keepaspectratio]{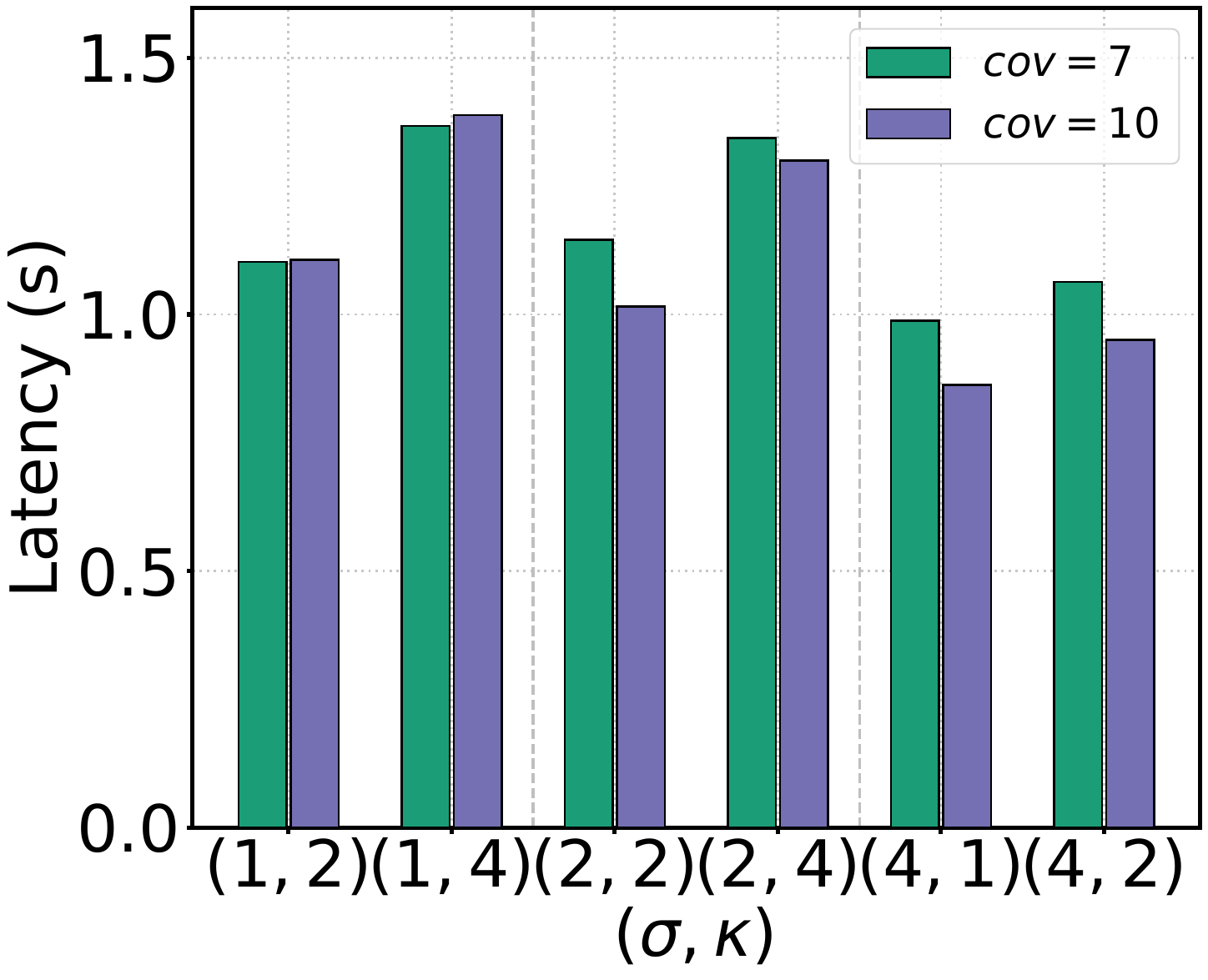}
        \caption{$\sigma$, $\kappa$ for diverse $\mathit{cov}$ under N1.}
        \label{fig:exp2-f}
    \end{subfigure}\hfill
    \begin{subfigure}[b]{0.25\linewidth}
        \centering
        \includegraphics[height=3.5cm, keepaspectratio]{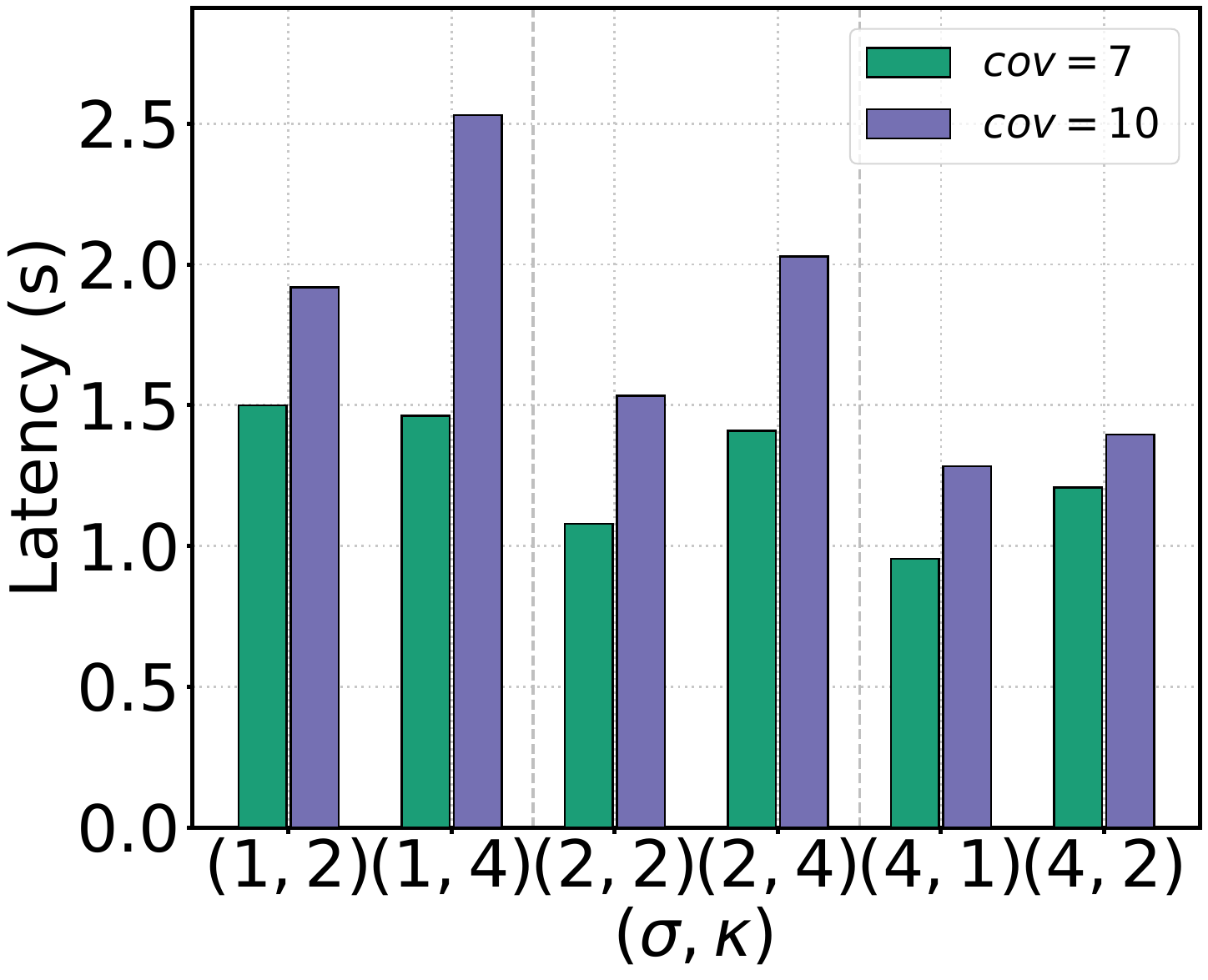}
        \caption{$\sigma$, $\kappa$ for diverse $\mathit{cov}$ under N2.}
        \label{fig:exp2-g}
    \end{subfigure}
    \caption{\yyu{Impact of framework parameters under benign and Byzantine settings. N1 and N2 represent homogeneous and geographically heterogeneous node deployments, respectively.}}
    \label{exp2}
    
    \end{minipage}%
    }
    \vspace{-1em}
\end{figure*}

\subsection{Implementation and Experimental Setup} \label{sec:analysis}

\bheading{Implementation.}
We implement \sysname on top of the public Tusk codebase~\cite{danezis2022narwhal}, which is written in Rust and uses Tokio~\cite{tokio} for asynchronous networking\footnote{The source code is at \url{https://github.com/michaelwcl323/manta-nsdi27}.}. 

The implementation uses \texttt{ed25519-dalek} for signature generation and verification, TCP for reliable peer-to-peer communication, and RocksDB~\cite{rocksdb_website} for persistent storage of consensus data. For fair comparison, we also implement four state-of-the-art asynchronous counterparts on the same platform, following their open-source designs. The key configurations are as follows.

\begin{packeditemize}
    \item \textbf{Tusk and DAG-Rider.} We instantiate these protocols by reconfiguring the common framework parameters. Tusk is obtained by setting $ref=2f+1$, while DAG-Rider is obtained by setting $\kappa=3$ and $ref=2f+1$.

    \item \textbf{Chitu.} Chitu is built on DAG-Rider with fast-path optimizations in the consensus layer. It is configured with $\kappa=3$ and $ref=2f+1$, and further incorporates an adaptive waiting mechanism to ensure that the fast path is triggered under normal conditions.

    \item \textbf{Mahi-Mahi.} Mahi-Mahi is configured with $\kappa=3$ and $ref=2f+1$. We also replace reliable broadcast with best-effort broadcast to construct an uncertified DAG, and re-implement the consensus layer accordingly.
\end{packeditemize}

\bheading{Experimental setup.} 
We deploy \sysname on the CloudLab platform~\cite{duplyakin2019design} using C6220 instances with 8 CPUs, 64\,GB RAM, and Ubuntu 22.04.2 LTS. To approximate realistic deployments, we refer to the validator distributions of Sui~\cite{suivision_validators} and Aptos~\cite{aptos_validators_explorer}, two representative DAG-based public blockchains, and construct a WAN with three clusters: Cluster A containing 60\% of nodes, Cluster B 30\%, and Cluster C 10\%. Inter-cluster latency and jitter are emulated using Linux \texttt{tc-netem}. The one-way delays are configured as follows: A--B: 51.3\,ms (1.0\,ms jitter), B--C: 76.6\,ms (1.5\,ms jitter), and A--C: 117.15\,ms (2.5\,ms jitter). We do not inject additional delay within clusters.

Due to limited availability of C6220 instances on CloudLab, experiments requiring 50 nodes are conducted on AWS EC2~\cite{aws} using \texttt{t3.2xlarge} instances with 8 vCPUs, 30\,GB memory, and Ubuntu 22.04.4 LTS. Nodes and clients are deployed across three regions: Frankfurt (\texttt{eu-central-1}, 60\%), Hong Kong (\texttt{ap-east-1}, 10\%), and Ohio (\texttt{us-east-2}, 30\%), with RTTs between regions ranging from 99\,ms to 203\,ms.

Each node is paired with a client that submits a fixed number of 512\,B transactions every 50\,ms. Nodes batch received transactions into newly generated vertices. Following standard blockchain benchmarking methodology, we define consensus latency as the time from when a transaction is first received by a node until it is successfully ordered. Unless otherwise stated, we use $(\sigma=2, ref=f+1, cov=2f+1, \kappa=2)$ as the default configuration of \sysname.

\begin{figure*}[!t]
    \centering
        \resizebox{1\textwidth}{!}{%
    \begin{minipage}{\textwidth}
    \centering
    
    \includegraphics[width=0.45\linewidth]{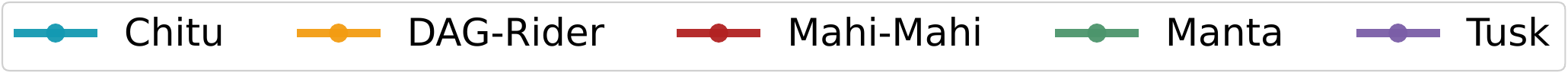}
    
    \begin{subfigure}{0.3\linewidth}
        \centering
        \includegraphics[width=\linewidth]{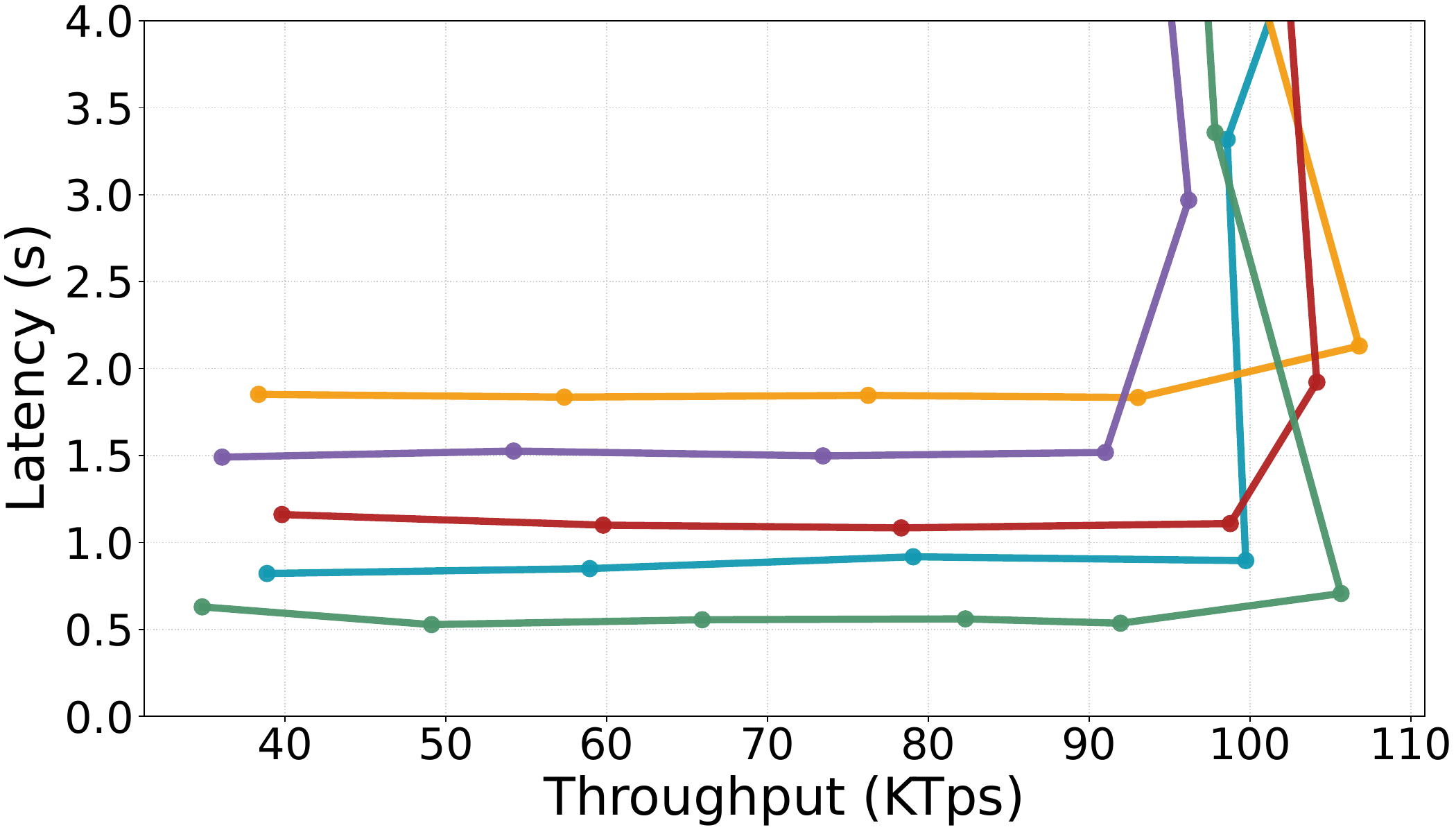}
        \caption{$n$=10, fault-free.}
        \label{fig:exp3n10}
    \end{subfigure}%
    \hspace{1em}
    \begin{subfigure}{0.3\linewidth}
        \centering
        \includegraphics[width=\linewidth]{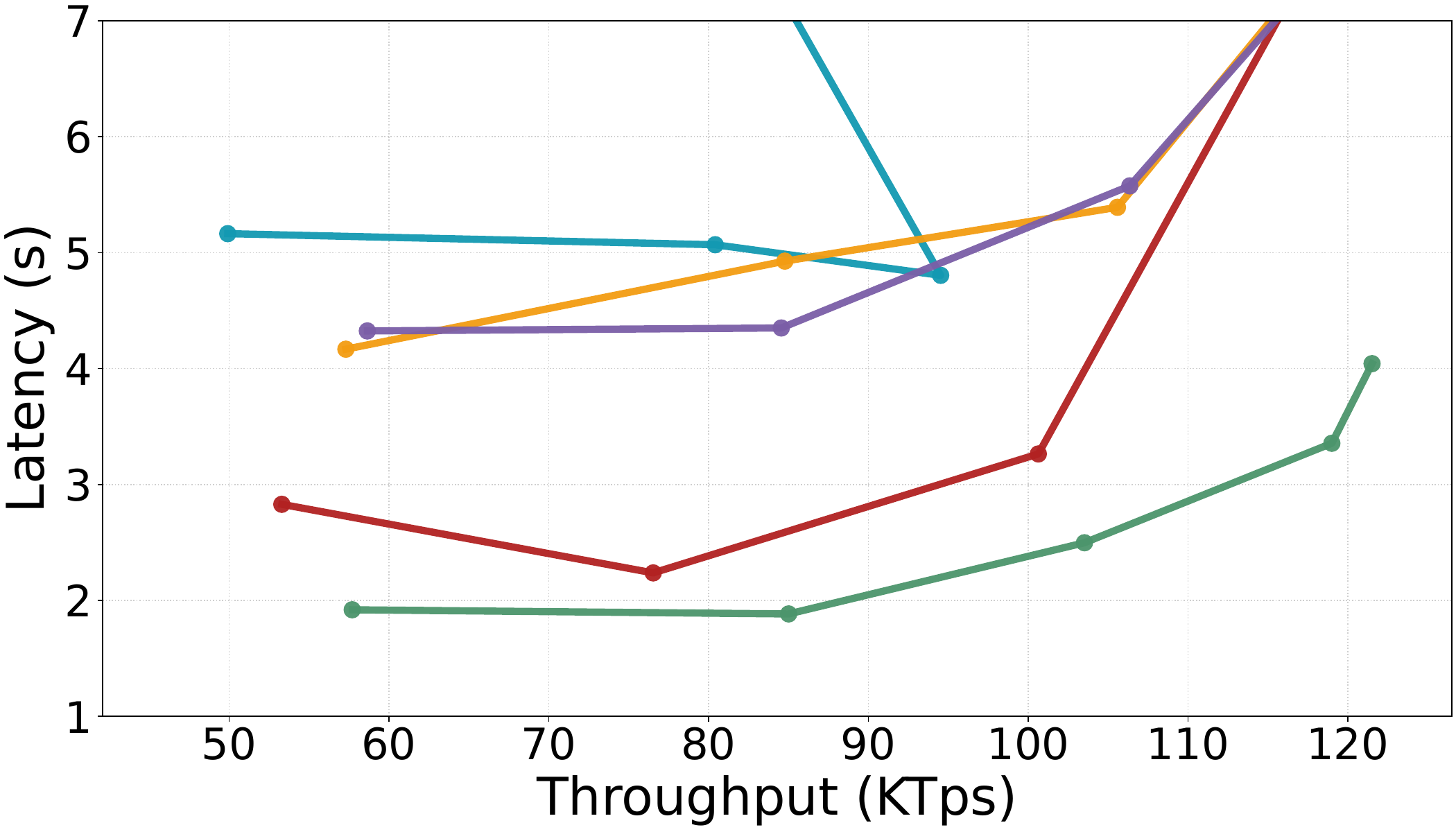}
        \caption{$n$=50, fault-free.}
        \label{fig:exp3n50}
    \end{subfigure}%
    \hspace{1em} 
    \begin{subfigure}{0.3\linewidth}
        \centering
        \includegraphics[width=\linewidth]{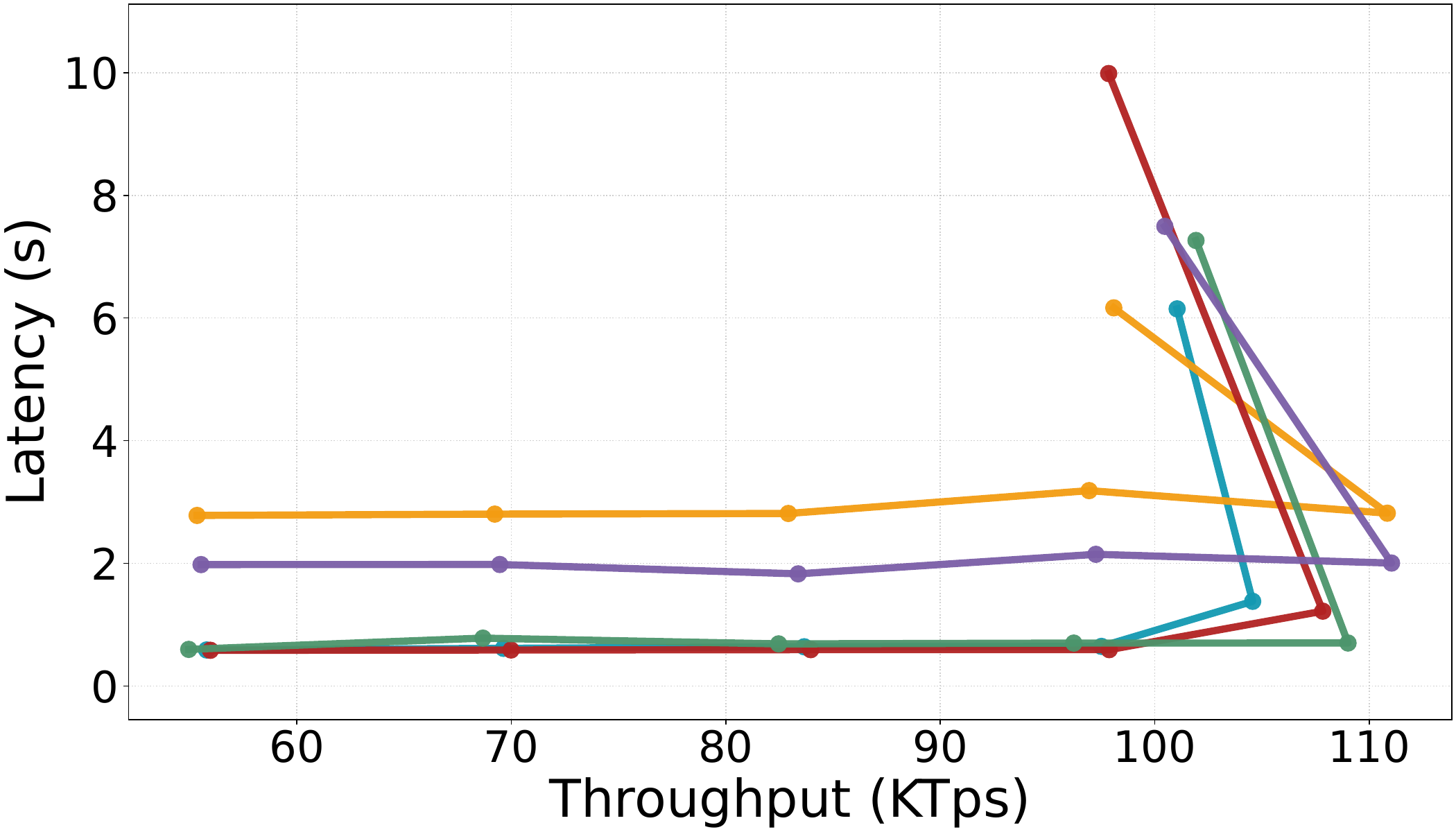}
        \caption{\yyu{$n$=10, 3 faults.}}
        \label{fig:exp3Byzantine}
    \end{subfigure}
    
    \caption{Throughput vs. latency in fault-free and \yyu{fault} settings.}
    \label{fig:exp3_main} 
        \end{minipage}%
    }
    \vspace{-1em}
\end{figure*}

\subsection{Design Parameter Trade-off} \label{subsec:theory-align}
We next evaluate the four key parameters $(\sigma, ref, cov, \kappa)$ in the framework \yyu{under both homogeneous and heterogeneous network conditions} and examine the trade-offs among them. The first two parameters shape DAG growth and progression, and therefore primarily affect latency. The latter two govern the robustness of the commit condition and are introduced to preserve safety and liveness in the worst case.

\bheading{The impact of $\sigma$ and $ref$.}
We fix $cov=2f+1$ and $\kappa=2$, and evaluate configurations with \yyu{$\sigma \in \{1,2,3,4,5,6\}$ and $ref \in \{2,4,7\}$} under $n=10$ and $f=3$. \yyu{Fig.~\ref{fig:exp2-b} shows that under network heterogeneity, $\sigma>1$ achieves lower latency than $\sigma=1$, as a larger span enables more flexible DAG progression and better exploits the flexible coin; a smaller $ref$ further reduces synchronization delay. In the homogeneous setting (Fig.~\ref{fig:exp2-a}), this advantage becomes smaller as inter-replica delay variation decreases. Latency remains similar for $\sigma>2$ because increasing the span does not introduce additional serial coverage dependencies once the solid step is formed.} 

\bheading{The impact of $\kappa$ and $cov$.}
\yyu{The parameters $cov$ and $\kappa$ govern worst-case robustness and the number of solid steps within a wave, respectively. We vary $cov \in \{2,4,7,10\}$ and $\kappa \in \{2,4,5\}$ under both network settings, first fixing $\sigma=1$, where $ref=cov$. Figs.~\ref{fig:exp2-d} and~\ref{fig:exp2-e} show that latency generally increases with $\kappa$, while at low $cov$, additional solid steps can improve accumulated connectivity and reduce latency. In the homogeneous setting, latency varies little with $cov$, whereas under network heterogeneity, a large $cov$ increases latency by requiring more slow replicas.} \yyu{We further vary $\sigma$ and $\kappa$ with $ref=4$ and $cov \in \{7,10\}$. For example, $(\sigma,\kappa)=(1,4),(2,2),(4,1)$ have the same number of rounds per wave but different internal constructions. Figs.~\ref{fig:exp2-f} and~\ref{fig:exp2-g} show that fewer, longer solid steps achieve lower latency than more, shorter solid steps. In the homogeneous setting, increasing $cov$ can further reduce latency by improving connectivity with little additional synchronization cost.}

We also evaluate a Byzantine scenario, where faulty replicas selectively delay messages to reduce shared DAG connectivity. Fig.~\ref{fig:exp3-kappa-Byzantine} shows that when $cov=f+1$, latency increases sharply under attack, \y{while the degradation gradually decreases as $cov$ and $\kappa$ increase and is nearly neutralized at $\kappa = 3$, $cov = 10$.} This is because larger $cov$ and $\kappa$ can yield a larger common core under this adversarial schedule, leaving less room for the attacker to disrupt commitment. 

\yyu{Based on these observations, in heterogeneous networks, the selected configuration of \sysname captures most of the flexibility benefit while satisfying correctness requirements.}

\subsection{Performance of \sysname} \label{subsec:comparison}
\subsubsection{Fault-Free Performance} \label{subsec:fault-free}
Figs.~\ref{fig:exp3n10} and~\ref{fig:exp3n50} compare \sysname against DAG-Rider, Tusk, Chitu, and Mahi-Mahi under fault-free settings (\ie, all nodes remain active throughout execution) at two system scales: $n=10$ and $n=50$ replicas. \sysname achieves the lowest latency at $n=10$, reducing latency by 30\%--70\%. Two mechanisms drive this gain. First, flexible DAG growth allows \sysname to progress with smaller quorums and form useful connectivity earlier, rather than waiting for rigid per-round thresholds. This is particularly beneficial in heterogeneous WANs with uneven cross-region delays. Second, the flexible coin enables commitment as soon as sufficient connectivity emerges, allowing earlier finalization. Although Chitu and Mahi-Mahi also partially mitigate rigid waiting through optimistic paths or lightweight dissemination, uneven delays can still enlarge divergence in local DAG views, making their commit conditions slower or more expensive to satisfy.

Fig.~\ref{fig:exp3n50} shows that the advantage of \sysname persists at a larger scale. It reduces latency by around 70\% compared with Chitu, Tusk, and DAG-Rider. Compared to Mahi-Mahi, the second-best baseline in this experiment, \sysname delivers around 30\% lower latency and 18\% higher throughput, demonstrating robust scalability in larger deployments.

\begin{figure}[!t]
    \centering
    \includegraphics[width=0.4\linewidth]{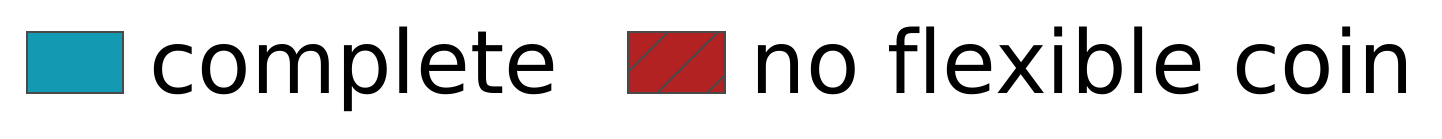} \\[0.1em]
 
    \begin{subfigure}{0.45\linewidth}
        \centering
        \includegraphics[width=\linewidth]{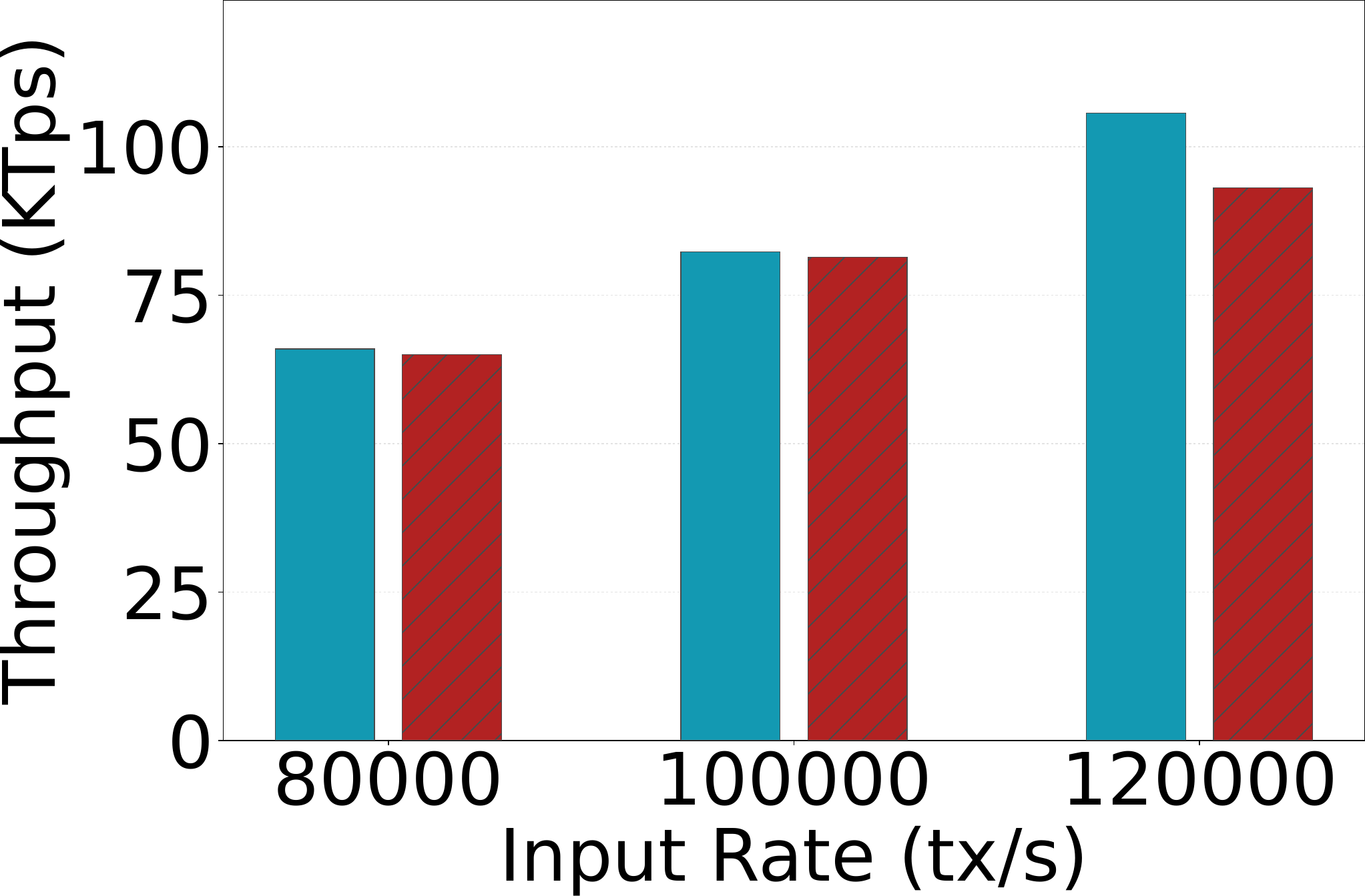}
        \caption{Throughput.}
        \label{fig:breakdowntps}
    \end{subfigure}%
    \hspace{0.5em}
    \begin{subfigure}{0.45\linewidth}
        \centering
        \includegraphics[width=\linewidth]{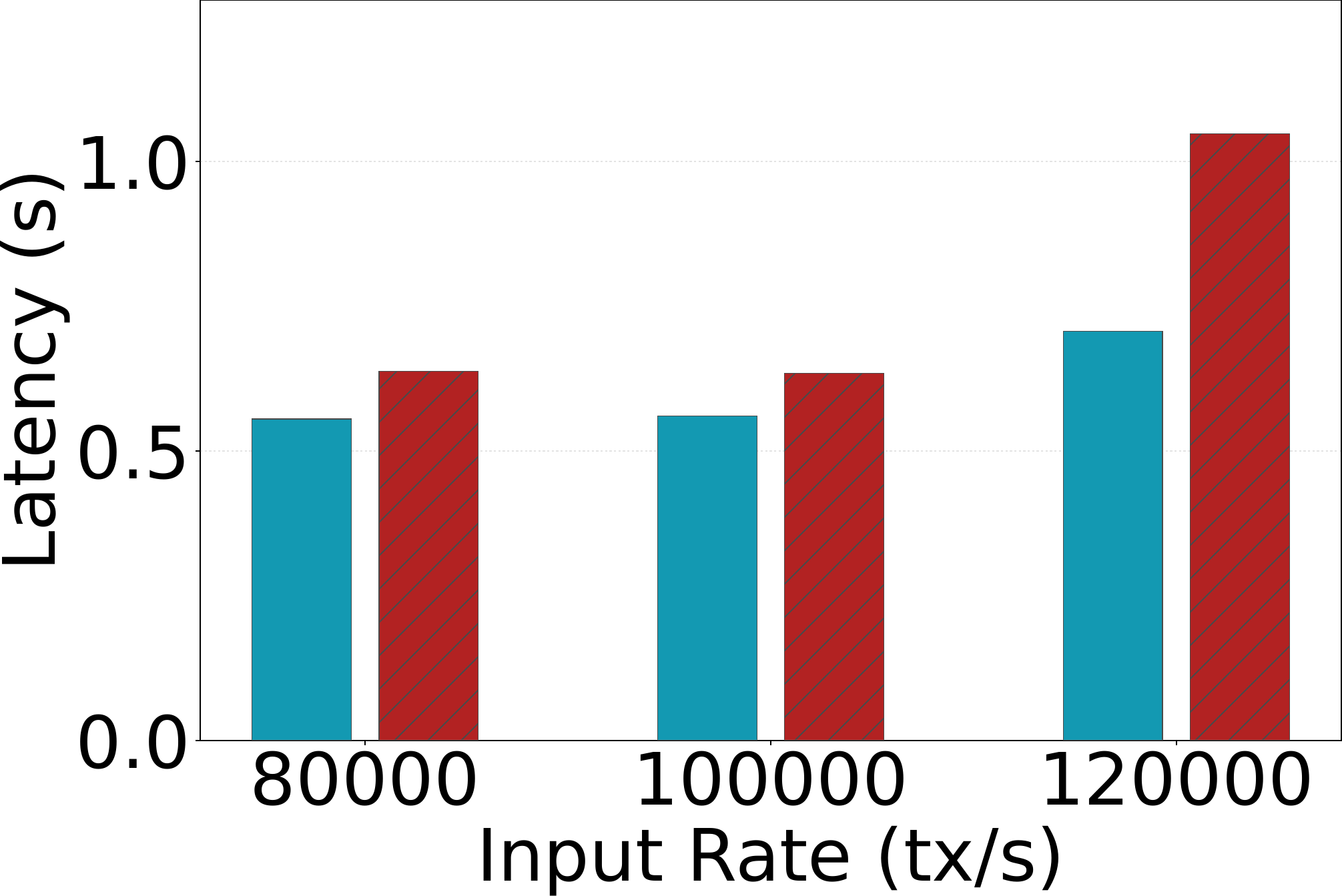}
        \caption{Latency.}
        \label{fig:breakdownlatency}
    \end{subfigure}%
    \caption{Ablation of flexible coin.}
    \label{fig:exp4_main} 
    \vspace{-1em}
\end{figure}

\subsubsection{Performance under Faults} \label{subsec:fault}
Fig.~\ref{fig:exp3Byzantine} evaluates the resiliency of \sysname under silent faults, where up to $f$ nodes are silent (\ie, crash or intentionally stop participating). We focus on silent faults because they directly stress DAG progression and are a common practical source of slowdown in WAN deployments. \yyu{As shown in Fig.~\ref{fig:exp3Byzantine}, \sysname shows performance comparable to the other protocols under faults. The remaining replicas can form more complete connectivity, potentially making it easier for baseline protocols to follow their fast paths and leaving less room for \sysname to benefit from faster progression.}

\vspace{-0.6em}
\subsection{Breakdown and Resource Utilization} 
\label{subsec:breakdown}

\yyu{Fig.~\ref{fig:exp4_main} compares \sysname with a variant that disables the flexible coin.} The results show that the flexible coin consistently reduces latency by enabling replicas to check commitment as soon as sufficient connectivity emerges, rather than waiting until the completion of a wave. Its benefit becomes more pronounced as the input rate increases. Under heavier workloads, faster replicas generate additional vertices while waiting to advance to the next solid step, \yyu{which can be committed by the flexible coin once they become sufficiently connected.}

\yyu{We further measure the effect of flexible DAG growth on round duration. Flexible DAG growth introduces more rounds within a wave but shortens each round from about 139.4 ms to 75.8 ms. The flexible coin complements this progression by allowing replicas to commit before the final round. Together, flexible DAG growth speeds up round progression, and the flexible coin converts this progress into earlier commitment.}

\begin{table}[!t]
\centering
\small
\caption{Resource utilization under system breakdown.}
\label{tab:hardware}
\renewcommand{\arraystretch}{0.95}
\setlength{\tabcolsep}{4pt}
\scalebox{1}{
\begin{tabular*}{\columnwidth}{@{\extracolsep{\fill}}lccc@{}}
\toprule
Variants & Complete & No flexible coin & Tusk \\
\midrule
CPU Avg. & 8.744\% & 8.831\% & 8.272\% \\
BW Avg.  & 45.908\% & 45.955\% & 38.313\% \\
\bottomrule
\end{tabular*}
}
\end{table}

\bheading{Resource Utilization.} Table~\ref{tab:hardware} reports the resource usage of \sysname. The flexible coin incurs negligible additional cost, with CPU utilization and bandwidth consumption remaining nearly unchanged compared with the variant without it. The complete protocol uses slightly more resources than Tusk, mainly due to the metadata required for solid-step tracking. \yyu{In addition, Boost edges introduce additional cross-round references and may incur DAG state management overhead. Overall, the average per-vertex metadata is 1,465.8 B for \sysname, compared with 388.0 B for Tusk.}

\bheading{\y{Limitations and Trade-offs.}}
\yyu{Beyond these results, several limitations and open directions remain. A lower reference threshold may disadvantage slow replicas, while $\tau_{\mathrm{cand}}$ introduces a trade-off between early recovery and worst-case latency under Byzantine behavior. Dynamic adaptation would require coordinated updates, \eg, through observation-driven mechanisms~\cite{de2026barnacle} or BFT reconfiguration~\cite{meng2024rondo}. Overall, these limitations motivate dynamic parameter adaptation and finer-grained coin scheduling as future work.}

\section{Related Work} \label{sec:related}
\bheading{Asynchronous BFT.} 
Asynchronous BFT protocols, such as HoneyBadgerBFT~\cite{miller2016honey} and the Dumbo family~\cite{guo2020dumbo,guo2022speeding}, avoid timing assumptions by leveraging randomization to achieve consensus in asynchronous environments and are therefore robust to WAN uncertainty. More recently, asynchronous DAG-based protocols~\cite{keidar2021all,danezis2022narwhal,huang2025chitu,jovanovic2025mahi} further improve performance through structured parallelism. DAG-Rider~\cite{keidar2021all} is the first to realize asynchronous BFT over a DAG. Tusk~\cite{danezis2022narwhal} extends this direction with a mempool layer that decouples metadata from data dissemination. Chitu~\cite{huang2025chitu} adds a fast path that directly commits under favorable conditions. 
Mahi-Mahi~\cite{jovanovic2025mahi} further improves performance through uncertified DAG construction and lightweight dissemination. However, these systems largely retain fixed DAG progression rules, whereas \sysname focuses on redesigning DAG progression itself to better exploit heterogeneity in WANs.

\bheading{Partially synchronous DAG-based BFT.} Another line of DAG-based protocols~\cite{babel2025mysticeti,shrestha2025sailfish,arun2025shoal++,spiegelman2022bullshark,giridharan2024autobahn,shrestha2025towards} assumes partial synchrony networks and uses designated leaders to accelerate commitment, often through optimistic paths. These protocols perform well when underlying timing assumptions hold, but can suffer performance degradation under highly heterogeneous WAN conditions~\cite{miller2016honey}. By contrast, \sysname operates under a fully asynchronous model without synchrony assumptions.

\bheading{Uncertified DAG-based BFT.} Some recent systems build DAGs using uncertified vertices~\cite{jovanovic2025mahi,babel2025mysticeti} and replace reliable broadcast with best-effort dissemination to reduce overhead. While efficient in favorable environments, such designs are more sensitive to packet loss and delay variation in WANs, where replicas may miss data and must fetch it later. The data recovery can block the critical path and degrade end-to-end performance. Accordingly, we focus on certified DAG constructions as prior designs~\cite{keidar2021all,danezis2022narwhal,huang2025chitu}.

\bheading{Flexibility in DAG-based BFT.} Several works introduce flexibility into DAG-based consensus. Angelfish~\cite{yu2025angelfish} adapts vertex and vote reference rules, while Sparse DAG~\cite{anoprenko2025dags} reduces metadata through sparse references. These designs explore different forms of flexibility, while \sysname focuses on relaxing round progression under asynchrony. Lemonshark~\cite{hu2026lemonshark} adopts a Narwhal-style DAG and uses key-space sharding for early finality, which is orthogonal to \sysname. \yyu{Avalanche~\cite{rocket2019scalable} avoids explicit round boundaries and progresses through probabilistic sub-sampling, with safety characterized through probabilistic convergence. Round-based protocols use explicit rounds and waves, simplifying safety analysis and commit reasoning but coupling progress to round boundaries. \sysname follows this model while relaxing round advancement conditions.}

\section{Conclusion} \label{sec:conclusion}
We revisit a core assumption in asynchronous DAG-based BFT consensus: connectivity must accumulate through rigid round-by-round progression. We show that this rigidity becomes a major bottleneck in heterogeneous WANs. We present a flexible DAG growth framework that decouples commit requirements from round advancement. Based on it, we design \sysname, which combines \textit{flexible growth} for earlier connectivity and a \textit{flexible coin} for earlier commitment. Experiments show that \sysname consistently outperforms prior asynchronous DAG-based BFT protocols \y{under fault-free heterogeneous deployments}, reducing latency by at least 30\% while preserving robustness. 

\clearpage
\normalem
\bibliographystyle{unsrt}
\bibliography{bib}

\appendix
\section{Building Blocks} \label{app:building blocks}

\subsection{Reliable Broadcast (RBC)} 
\label{subsec: RBC}

RBC~\cite{bracha1987asynchronous} is a fundamental primitive in consensus protocol designed to ensure non-equivocation. It guarantees that even under a Byzantine sender, all honest parties eventually deliver the same message or none at all. Formally, each replica $p_k$ can broadcast messages by calling $Bcast(m,r)$, where $m$ is the message, and $r \in \mathbb{N}$ is a round number. By running the protocol, every $p_i$ has an output $Deliver(m,r,p_k)$, where $m$ is the message, $r$ is a round number, and $p_k$ is the replica that invoked the corresponding $Bcast(m,r)$. RBC typically satisfies the following properties:

\bheading{Agreement.} If an honest replica $p_i$ outputs $Deliver(m,r,p_k)$, then every other honest replica $p_j$ eventually outputs $Deliver(m,r,p_k)$.

\bheading{Integrity.} For each round $r$, an honest replica $p_i$ outputs $Deliver(m,r,p_k)$ at most once regardless of $m$.

\bheading{Validity.} If an honest replica $p_k$ calls $Bcast(m,r)$, then every honest replica $p_i$ eventually outputs $Deliver(m,r,p_k)$.

\subsection{Global Perfect Coin} \label{subsec: coin}
The global perfect coin~\cite{libert2014born,boneh2001short,miller2016honey} (also known as common coin) is a cryptographic primitive that allows a set of distributed nodes to reach a consensus on a common, unpredictable random value in an asynchronous network. By providing shared randomness, it serves as a critical mechanism to circumvent the FLP impossibility~\cite{fischer1985impossibility}. For each instance $w \in \mathbb{N}$, any replica $p_i \in \Pi$ can call $ChooseLeader_i(w)$ to identify the leader for that instance. The coin ensures that all honest replicas receive the same leader $p_j \in \Pi$, and this outcome remains unknown to the adversary until the call is made. We denote $X_w$ as the random variable representing the leader selection outcome for instance $w$. The global perfect coin provides the following guarantees:

\bheading{Agreement.} If two honest replicas $p_i, p_j \in \Pi$ invoke $ChooseLeader_i(w)$ and $ChooseLeader_j(w)$ respectively, they obtain the same output, \ie, $p^*$. 

\bheading{Termination.} If at least $f+1$ replicas invoke $ChooseLeader(w)$, then every invocation by an honest replica eventually returns a value.

\bheading{Unpredictability.} As long as fewer than $f+1$ replicas have invoked $ChooseLeader(w)$, the outcome $X_w$ remains indistinguishable from a uniform random distribution to the adversary, except with negligible probability $\epsilon$.

\bheading{Fairness.} The coin is unbiased such that for every instance $w \in \mathbb{N}$, $\forall p_j \in \Pi, \Pr[X_w = p_j] = 1/n$.

\subsection{Common Core Abstraction}
We follow the common core abstraction used in prior works~\cite{abraham2023perfectly,dolev2016some}. Each replica $p_i$ starts with an input $x_i$, and outputs a set of received values. Let $V_i^t$ denote the set held by replica $p_i$ after round $t$, for $t= 1,2,...,k$. The common core protocol proceeds in k rounds, with each round aggregating the sets/inputs received from the preceding round: 

\noindent\textbf{Round $1$.} Each replica broadcasts its input $x_i$ using reliable broadcast. $V_i^1 = \{x_i\}$.

\noindent \textbf{Round $2$.} Upon collecting at least $ref$ inputs $x_j$ from round 1, $p_i$ unions them into a set $V_i^2 = \bigcup V_j^{1}$ and sends it to all replicas.

\noindent \textbf{Round $t$ ($t\geq 3$).} Upon collecting at least $ref$ sets from round $t-1$, $p_i$ unions them into a set $V_i^t = \bigcup V_j^{t-1}$, and sends it to all replicas.

\noindent\textbf{Output.} After completing round $k$, each replica outputs $V_i^k$.

The key property of this abstraction is that there exists a set $S$ such that $S \subseteq V_i^k$ for every honest replica $p_i$. We refer to such a set $S$ as a common core. The size of the common core, \ie, $|S|$, captures how many first-round input values are guaranteed to be commonly visible at the end of the protocol. As we mentioned in Sec. \ref{subsec:problem_statement}, $|S|$ determines the commitment probability of a wave. We model the latency under the following three scenarios (Heterogeneous Network \ref{app: hetero modeling}, Random-delay Network \ref{app: random modeling}, and Adversarial Network \ref{subsec: worstcase}).

\section{DAG Growth Proof} \label{app:modeling} 
\subsection{Heterogeneous Network Model} \label{app: hetero modeling}

For analytical tractability, we assume a stable aggregate arrival trend across replicas. Vertices can be viewed as arriving in tiers, where faster tiers reach the reference quorum and become commonly visible earlier than slower tiers. \y{This model abstracts DAG evolution into a shared arrival process, using independent quorum delays without tracking replica-specific views or cross-round causal dependencies. Despite these simplifications, the model captures the key trends in latency and connectivity as the network delay distribution and DAG parameters vary.}

\paragraph{Assumption 1 (Network delay).} Let $d$ denote the point-to-point network delay between two replicas. 
We assume that $d$ follows a Log-normal distribution:

$$ \ln d \sim \mathcal{N}(\mu,std^2)$$ 

Let $f(x)$ be the probability density function (PDF) and $F(x) = P(d \le x)$ be the cumulative distribution function (CDF), representing the probability that a vertex from a random replica $i$ reaches the local replica within time $x$.

\begin{definition} [Round Duration] Consider a system with $n$ replicas. A round can advance only after collecting vertices from at least $ref$ distinct replicas. Let $d_{(ref)}$ denote the $ref$-th order statistic of $n$ i.i.d. copies of $d$. Then $d_{(ref)}$ represents the time required to complete one logical round.
\end{definition}

\begin{lemma}[Distribution of Round Duration]
The CDF of the round duration $d_{(ref)}$ is
$$F_{d_{(ref)}}(x) = \sum_{j=ref}^{n} \binom{n}{j} [F(x)]^j [1-F(x)]^{n-j}$$
\end{lemma}

\begin{proof}
The event $d_{(ref)} \le x$ is equivalent to saying that at least $ref$ among the $n$ delays are no larger than $x$. Since each delay is at most $x$ with probability $F(x)$, the result follows from the Binomial distribution.
\end{proof}

\begin{definition} [Single-round window]
    For a target confidence level $c \in (0,1)$, define
$$W_{\mathrm{step}}(c)=\inf \{t : F_{d_{(ref)}}(t) \ge c\}.$$
This is the minimum physical time within which one round completes with confidence at least $c$.
\end{definition}

\begin{definition}[Wave window]
For a wave consisting of $k$ rounds, let
$$
H_{k-1}=\sum_{i=1}^{k-1} d_{(ref)}^{(i)},
$$
where $d_{(ref)}^{(1)}, \dots, d_{(ref)}^{(k-1)}$ are independent copies of $d_{(ref)}$. Intuitively, a vertex that is observed at the decision point in round $k$ has undergone propagation for $k-1$ rounds. Accordingly, the distribution of $H_{k-1}$ is determined by the $(k-1)$-fold convolution of the distribution of $d_{(ref)}$. For a target confidence level $c \in (0,1)$, we define the wave window as
$$
W(ref,k-1)=\inf\{t: F_{H_{k-1}}(t)\ge c\},
$$
where $F_{H_{k-1}}$ denotes the cumulative distribution function of $H_{k-1}$.
\end{definition}

\begin{definition}[Effective visible set size]
    For any physical time $t$, define
$$\text{Count}(t) = \max \{ j \in \{1, \dots, n\} : F_{d_{(j)}}(t) \ge c \}$$

which represents the number of replicas whose vertices are visible within time $t$ with confidence at least $c$. At the end of round $k$, define the effective visible set size as
$$\Phi(ref, k) = \max \left( ref, \text{Count}(W(ref,k-1)) \right)$$
\end{definition}

\begin{theorem}[Commit probability under heterogeneous delay] \label{thm:commit-prob}
Under the above model, the number of effectively visible first-round vertices after a wave is lower bounded by $\Phi(ref, k)$. If the leader is selected uniformly at random from the $n$ first-round vertices, then the commit probability admits the lower bound
$$P\ge \frac{\Phi(ref, k)}{n}.$$

\end{theorem}
\begin{proof}
By construction, after the wave window $W(ref,k-1)$, at least $\Phi(ref, k)$ first-round vertices are visible with confidence $c$. Under the conservative assumption that only these visible vertices are safely committable, and the leader is chosen uniformly from all $n$ candidates, the probability that the chosen leader belongs to this visible set is at least $\Phi(ref, k)/n$.
\end{proof}

\begin{theorem}[Latency under heterogeneous delay]
Let $W(ref,k)$ denote the $c$-quantile of the \textit{wave duration} covering $k$ rounds. Under the above model, the protocol latency is approximated by
$$\mathrm{L}
\approx
\frac{W(ref,k)}{P}.$$
Using the lower bound on $P$ from Theorem~\ref{thm:commit-prob}, this yields a latency estimate in terms of $W(ref,k)$ and $\Phi(ref, k)$.
\end{theorem}

\subsection{Random-Delay Network Model (Stochastic case).} \label{app: random modeling}
In the random network model, we assume that in each round $r$, an honest replica references a uniformly random subset of $ref$ replicas from the preceding round $r-1$. 

Let $x_r \in \{0, \dots, n\}$ be the number of vertices in round $r$ who possess a causal path to the leader $v$. The growth of leader visibility depends on the probability $q_r$ that a vertex in round $r$ hits at least one successful vertex (\ie, those with a path to $v$) from the preceding round. We define $q_r$ as a piecewise function:

$$q_r = \begin{cases} \mathit{ref}/n, & r=2 \\ 1 - \binom{n-x_{r-1}}{\mathit{ref}} \big/ \binom{n}{\mathit{ref}}, & r > 2 \end{cases}$$

To formalize the recursive transition, we define the Propagation Operator $\mathcal{B}_r$ as the binomial probability mass function:
$$\mathcal{B}_r = \text{Bin}(x_r; N_r, q_r)$$ where the trial size $N_r$ is $n$ for intermediate propagation rounds ($r < k$) and $\mathit{ref}$ for the final decision round ($r = k$).

A leader is committed if its support in the final decision round $k$ satisfies the threshold $x_k \ge f+1$. Using the law of total probability, the commit probability $P$ is derived by marginalizing over all possible intermediate states:

1)  Two-round wave ($k=2$): In a 2-round wave, round 2 is the decision round. Each vertex in round 2 has $ref$ parents, and there is exactly 1 leader in round 1. The probability that a round-2 vertex connects to the leader is $q_2= ref/n$: 
$$P(2) = \sum_{x_2=f+1}^{\mathit{ref}} \mathcal{B}_2$$

2) Three-round ($k=3$): In a 3-round wave, round 2 acts as a propagation layer, and the decision is made in round 3 based on the accumulated visibility $x_2$:

$$P(3) = \sum_{x_2=0}^{n} \left[\mathcal{B}_2 \sum_{x_3=f+1}^{\mathit{ref}} \mathcal{B}_3 \right]$$

3) Generalized Multi-round Model ($k \ge 4$): For a wave of length $k$, the leader propagates through $k-2$ intermediate rounds before the final quorum check:
$$P(k)
=\sum_{x_2=0}^{n} \mathcal{B}_2 \sum_{x_3=0}^{n} \mathcal{B}_3\cdots\sum_{x_{k-1}=0}^{n} \mathcal{B}_{k-1}\sum_{x_k=f+1}^{\mathit{ref}} \mathcal{B}_k
$$

In this random delay model, we quantify the expected latency by the number of rounds required for a vertex to be committed, expressed as:
$$
E(rounds) = \frac{k}{P(k)}
$$

We evaluate the commit probability $P(k)$ and the expected rounds $E[rounds]$ required across various $f$ under two typical configurations, as summarized in Table \ref{tab:P_commit} and Table \ref{tab:E_rounds}.

\begin{table}[h]
\centering
\caption{$P(k)$ for different $ref$ configurations.}
\label{tab:P_commit}
\begin{tabular}{ccccc}
\toprule
$ref$ & $f$ & $k=2$ & $k=3$ & $k=4$ \\
\midrule
\multirow{3}{*}{$2f+1$}
 & 1  & 0.8437 & 0.9887 & 0.9955 \\
 & 3  & 0.8739 & 0.9999 & 0.9999 \\
 & 10 & 0.9555 & 0.9999 & 0.9999 \\
\midrule
\multirow{3}{*}{$f+1$}
 & 1  & 0.2500 & 0.6354 & 0.8284 \\
 & 3  & 0.0256 & 0.6614 & 0.9744 \\
 & 10 & $\approx$ 0      & 0.9279 & 0.9999 \\
\bottomrule
\end{tabular}
\end{table}

\begin{table}[h]
\centering
\caption{$E(\textit{rounds})$ for different $ref$ configurations.}
\label{tab:E_rounds}
\begin{tabular}{ccccc}
\toprule
$ref$ & $f$ & $k=2$ & $k=3$ & $k=4$ \\
\midrule
\multirow{3}{*}{$2f+1$}
 & 1  & 2.37   & 3.03 & 4.01 \\
 & 3  & 2.28   & 3.00 & 4.00 \\
 & 10 & 2.09   & 3.00 & 4.00 \\
\midrule
\multirow{3}{*}{$f+1$}
 & 1  & 8.00         & 4.72 & 4.82 \\
 & 3  & 78.12     & 4.53 & 4.10 \\
 & 10 & --     & 3.23 & 4.00 \\
\bottomrule
\end{tabular}
\end{table}

\subsection{Adversarial Network Model (Worst case)}\label{subsec: worstcase}
In the adversarial network model, an asynchronous scheduler may manipulate the DAG topology to minimize the visibility of vertices. We define $n = 3f+1$ as the total number of replicas, and $ref$ as the number of references per vertex; the number of vertices in round $i$ is defined as $n_i$. A vertex $v \in R_1$ is committable if it is reachable from at least $f+1$ distinct vertices in the decision round $R_k$. 

To evaluate the liveness of the protocol under adversarial scheduling, we categorize the vertices in the first round ($R_1$) into two disjoint sets based on their guaranteed visibility from the decision round $R_k$:

\begin{packeditemize}
    \item Guaranteed set ($G_k$): the set of vertices in $R_1$ that are guaranteed to satisfy the commit rule under any admissible adversarial DAG topology; let $g_k = |G_k|$ denote its size.
    \item Uncertain set ($U_k$): the set of remaining vertices in $R_1$, whose committability depends on the realized topology and is therefore not guaranteed in the worst case. For the purpose of establishing a lower bound, we conservatively treat these vertices as non-committable.
\end{packeditemize}

\begin{lemma}
    Under the adversarial model with $n=3f+1$ and a wave length of $k=2$, if each vertex references $ref \ge n-f$ parents, then there are at least $f+1$ vertices in $R_1$ guaranteed to satisfy the commit rule.
\end{lemma}\label{app: lmk2}

\begin{proof}
Given that each vertex has a minimum reference count $ref$ and $n_1$ ranges within $[2f+1, 3f+1]$. Consider any set of $ref$ vertices in $R_2$, the total number of links from the set to the first round vertices is at least $ref \cdot ref$. If a vertex is in $G_k$, it receives $\ge f+1$ links, while a vertex in $U_k$ receives $\le f$ links. The maximum number of links from vertices in the set to each vertex in the first round is $ref$; therefore, a vertex in $G_k$ receives at most $ref$ links, while a vertex in $U_k$ receives at most $f$ links. Therefore, the total link balance is expressed as $g_k \cdot ref + (n_1 - g_k) \cdot f \ge ref \cdot ref$. To find the minimum bound, we set $n_1 = 3f+1$. Substituting $ref \ge n-f$:

$$g_k \ge \frac{f^2 + 3f + 1}{f+1} \ge f+1$$

This completes the proof. Similarly, we observe how $g_k$ varies with different values of $ref$ under $k=2$:

1) If $ref \ge f+1$, then $g_k \ge -2f^2+f+1 \le 0$ for $f \ge 1$, providing no liveness guarantee.

2) If $ref \ge 1.8f + 1$, then $g_k \ge 1$, which is the minimum requirement for liveness.

\end{proof}

\begin{lemma}
    Under the adversarial model with $n=3f+1$ and a wave length of $k=3$, if each vertex references $ref \ge n-f$ parents, then there are at least $1.5f$ vertices in $R_1$ guaranteed to satisfy the commit rule. \label{app: lmk3}
\end{lemma}

\begin{proof}
    Given that each vertex has a minimum reference count $ref$ and the round sizes $n_1, n_2$ range within $[2f+1, 3f+1]$, the total number of links from $R_2$ to $R_1$ is $n_2 \cdot ref$. In a wave of length $k=3$, a vertex $v \in R_1$ belongs to the set $G_k$ if it receives $\ge n_2 - ref + 1$ links from $R_2$. Since each vertex in $R_3$ references at least $ref$ parents from $R_2$, the principle of quorum intersection dictates that any selection of $ref$ vertices in $R_2$ must include at least one supporter of $v$. This ensures the existence of a causal path from every vertex in $R_3$ to $v$. Conversely, if a vertex receives $\le n_2 - ref$ links, it is classified into the set $U_k$. A vertex in $G_k$ receives at most $n_2$ links, while a vertex in $U_k$ receives at most $n_2 - ref$ links. Therefore, the total link balance is expressed as $g_k \cdot n_2 + (n_1 - g_k) \cdot (n_2 - ref) \ge n_2 \cdot ref$. Solving for $g_k$ yields the general bound: 
    $$g_k \ge n_1 + n_2 - \frac{n_1 \cdot n_2}{ref}$$ 
    
    To find the minimum bound, we set $n_1 = n_2 = 3f+1$. Substituting $ref \ge n-f$:
    $$
    g_k\ge \frac{(3f+1)(f+1)}{2f+1} \ge 1.5f
    $$
This completes the proof. Similarly, we observe how $g_k$ varies with different values of $ref$ under $k=3$:

1) If $ref \ge f+1$, then $g_k \ge \frac{(3f+1)(1-f)}{f+1} \le 0$ for $f \ge 1$, providing no liveness guarantee.

2) If $ref \ge 1.5f + 1$, then $g_k \ge 1$, which is the minimum requirement for liveness.
    
3) If $ref \ge 1.8f + 1$, then $g_k \ge f+1$.
\end{proof}

\begin{lemma}
    Under the adversarial model with $n=3f+1$ and a wave length of $k=4$, if each vertex references $ref \ge n-f$ parents, then there are at least $2f+1$ vertices in $R_1$ guaranteed to satisfy the commit rule.
\end{lemma}\label{app: lmk4}

\begin{proof}
We treat vertices in $R_1$ as the input to $R_2$, and the subsequent process from $R_2$ to $R_4$ as a wave of length $k=3$. According to Lemma \ref{app: lmk3}, when $ref \ge n-f$, there are at least $g_k \ge 1.5f$ vertices in $R_2$ that are guaranteed to be part of the common core. Let this set of vertices be $S \subset R_2$, where $|S| \ge 1.5f$. Each vertex $v \in S$ references at least $ref$ parents from $R_1$. To find the minimum number of vertices in $R_1$ that are included in the common core, we consider the worst-case scenario where the reference sets of all vertices in $S$ overlap maximally. Even if all $1.5f$ vertices in $S$ reference the exact same set of parents in $R_1$ to minimize the total count, they collectively guarantee the availability of at least $ref$ vertices. Given $ref = 2f+1$, this ensures that at least $2f+1$ vertices in $R_1$ are causally connected to the common core of the $k=4$ wave. This completes the proof.
\end{proof}

\section{\sysname Correctness Analysis} \label{app:correctproof}

\subsection{Safety Proof}
\begin{lemma}\label{lm: safety1}
    If an honest replica commits a leader vertex $v$ in a wave $i$, then any leader vertex $v'$ committed by any honest replica $p$ in a future wave has a path to $v$.
\end{lemma}

\begin{proof}
 Let $r_j$ denote the j-th round relative to the considered wave. In the standard process, a leader $v_l$ of wave $i$ is eligible for commitment if it is observed by at least $f+1$ vertices at $r_3$. Under the flexible coin mechanism, replicas may observe and commit $v_l$ from $r_2$, $r_3$, or $r_4$.
    
    \noindent\textbf{Case 1: Observation at $r_3$ or $r_4$.} We treat both as observation at $r_3$. Since the $r_3$ to $r_5$ interval constitutes a solid step, every vertex at $r_5$ (the first round of wave $i+1$) must have at least $2f+1$ links to $r_3$. By quorum intersection ($n=3f+1$), these $2f+1$ referenced vertices and the $f+1$ observers at $r_3$ share at least one common vertex. Thus, every vertex in the first round of wave $i+1$ has a path to $v_l$. 
    
    \noindent\textbf{Case 2: Observation at $r_2$.} When the $r_2$-$r_5$ solid step is active, the protocol allows observation at $r_2$. Following the same logic, every vertex at $r_5$ has at least $2f+1$ links to $r_2$. Quorum intersection ensures that every vertex in the first round of wave $w+1$ has a path to $v_l$.
    
    By a simple inductive argument, every vertex in every round of any wave higher than $w$ must contain a path to $v_l$.
\end{proof}

\begin{theorem} \label{lm: safety2}
    Any two honest replicas commit the same sequence of vertex leaders.
\end{theorem}

\begin{proof}
    Suppose an honest replica $p_i$ commits leader vertex $v$ of wave $w$. Another honest replica $p_j$ commits leader vertex $v'$ of wave $w'$, where $w < w'$, but has not yet committed $v$. According to Lemma \ref{lm: safety1}, there must exist a path from $v'$ to $v$. This implies that $v$ is already in the causal history of $v'$ and must be committed by $p_j$ indirectly. This contradicts the assumption that $p_j$ could commit $v'$ while skipping $v$. Therefore, if a leader vertex is committed by one honest replica, no other replica can skip it, ensuring all honest replicas commit leader vertices in the same sequence.
\end{proof}

\subsection{Liveness Proof}
\begin{lemma} \label{lm: liveness1}
    For every wave $w$, there are at least f + 1 vertices in the first round of $w$ that satisfy the commit rule.
\end{lemma}

\begin{proof}
    In the regular process, consider an arbitrary set of $2f+1$ vertices at round 3 in wave $w$. According to the solid step rule, each vertex must reference at least $2f+1$ vertices from round 1, resulting in a total number of links from round 3 to round 1 of at least $(2f+1)(2f+1)$. There are at most $n = 3f+1$ vertices in round 1. Let $h$ be the number of vertices that satisfy the commitment condition (receiving $\ge f+1$ links) and $l$ be the number of vertices that do not. To find the minimum $h$, we construct a worst-case scenario where each non-committable vertex receives the maximum possible $f$ links, while each committable vertex receives up to $2f+1$ links. The total number of links satisfies $h(2f+1) + l \cdot f \ge (2f+1)^2$, where $l = n - h$. Therefore, we have $h \ge f+1$, proving that at least $f+1$ vertices at round 1 satisfy the commitment condition.
\end{proof}

\begin{theorem}\label{lm: liveness2}
    Every transaction issued by a correct client is eventually committed.
\end{theorem}

\begin{proof}
    Lemma \ref{lm: liveness1} ensures that at least $f+1$ vertices satisfy the commitment condition. When the flexible coin is employed, its prerequisite is that a replica has already locally observed at least $\tau_{\mathit{cand}}$ vertices from the first round that meet the commitment criteria; thus, committable vertices are guaranteed to exist. While $\tau_{\mathit{cand}}$ can be configured based on the system's optimism, any $\tau_{\mathit{cand}} \ge 1$ is sufficient to guarantee liveness. This completes the proof.
\end{proof}

\section{Optimizations} \label{app:dag structure}
\sysname introduces a flexible DAG growth mechanism to provide significant latency advantages in heterogeneous networks. Beyond this, however, high network heterogeneity further exacerbates other inherent inefficiencies in traditional DAG-based BFT designs. Specifically, inconsistent replica speeds often cause slower replicas to block faster ones, leading to substantial bandwidth underutilization. Furthermore, as illustrated in \figref{fig:data1}, transactions not promptly included in block 2 must wait for the subsequent proposal block 3, which increases queuing delays. While the flexible growth mechanism alleviates these issues by relaxing round advancement, a more fundamental solution is required to eliminate these barriers. To this end, we propose a flexible data structure specifically tailored for heterogeneous environments, which further decouples transaction dissemination from consensus metadata. The data structure is categorized into two types:

\bheading{$T_{block}$}: Transaction Blocks, represented as squares in Fig. \ref{fig:data2}, contain the raw transaction data and are disseminated via best-effort broadcast.

\bheading{Vertex}: Represented as circles in Fig. \ref{fig:data2}, these contain the digest of a $T_{block}$ and its Proof of Verifiable Availability (PVA), disseminated via Byzantine Reliable Broadcast (RBC).

\begin{figure}[!t]
    \centering
    \begin{subfigure}{0.48\linewidth}
        \centering
        \includegraphics[width=\linewidth]{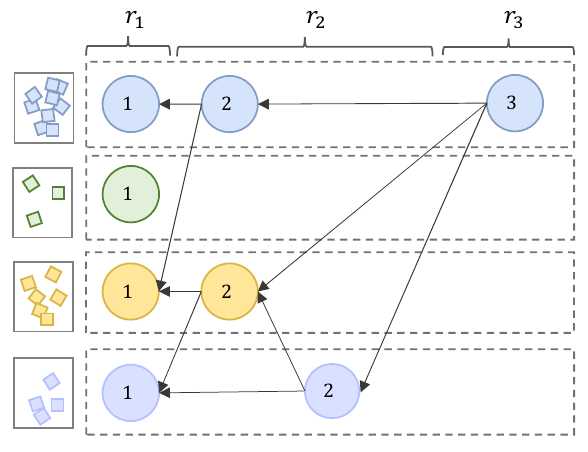}
        \caption{Rigid data integration.}
        \label{fig:data1}
    \end{subfigure}
    \hfill
    \begin{subfigure}{0.48\linewidth}
        \centering
        \includegraphics[width=\linewidth]{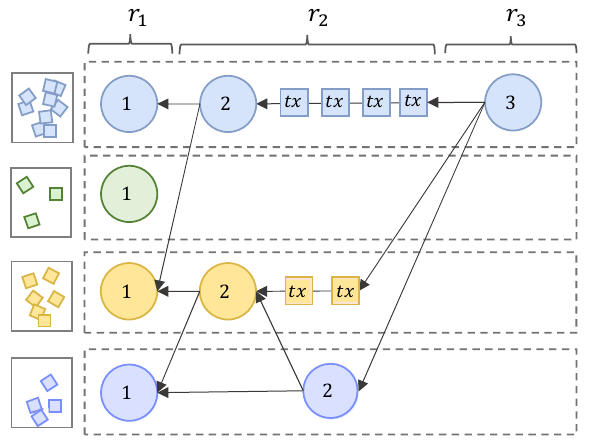}
        \caption{Flexible data allocation.}
        \label{fig:data2}
    \end{subfigure}
    \caption{Flexible data structure in $\sysname$.}
    \label{fig:data_structure}
\end{figure}

This design improves bandwidth utilization across replicas. A $T_{block}$ can be proposed immediately after broadcast, without waiting for the next vertex window or a designated proposer. To illustrate, consider a transaction arriving at Replica 3 (highlighted in orange in \figref{fig:data_structure}). Replica 3 packages it into a $T_{block}$ and broadcasts it. After collecting $f+1$ acknowledgments (ACKs), it aggregates them into a PVA and re-disseminates it. Once other replicas receive both the $T_{block}$ and its PVA, they can include the $T_{block}$ digest in their own vertices. For example, vertex Blue 3 incorporates two orange $T_{block}$s into the DAG.

Without this decoupled design, these $T_{block}$s would have to wait until Replica 3 issues its designated vertex (Orange 3) before entering the DAG. In WAN settings, if Orange 3 is delayed by transient network fluctuations, the $T_{block}$s incur substantial queuing delay. With our flexible construction, they join the DAG as soon as they are authenticated, reducing the chance of excessive delays caused by straggling replicas or unstable networks.

\bheading{Proof of Verifiable Availability (PVA).}
The replica broadcasts $T_{block}$ to all other replicas, and replicas vote to acknowledge delivery (ACKs). $f + 1$ ACKs constitute a PVA that guarantees that $T_{block}$ has been delivered by at least one honest replica. If a replica receives $T_{block}$ and its PVA, the replica can include the $T_{block}$ in its vertex (by including the digest and PVA of the $T_{block}$). By using PVA, if a vertex contains valid PVAs for each referenced $T_{block}$, it can be voted for in the RBC phase without waiting for $T_{block}$ contents. All missing $T_{block}$ can be fetched in the background.

\bheading{Reliable Broadcast (RBC)}. In \sysname, vertices reference the certified vertex from the preceding round. Each replica proposes and certifies at most one candidate per round. To certify a vertex, a replica follows a simplified RBC:
\begin{packedenumerate}
    \item To initiate the broadcast, a replica constructs and signs a vertex $v$ as its proposal. This vertex incorporates two key components: (1) a collection of $T_{block}$ digests along with their corresponding PVA harvested from the local mempool, and (2) several vertex certificates from preceding rounds that serve as boost or progress edges.
    \item Upon receiving a vertex from another peer, a replica first performs an equivocation check to ensure it is the first proposal received from that specific sender in the current round; if it is not, the vertex is discarded. If valid, the replica further verifies the included PVAs. Once the vertex passes all checks, the replica broadcasts a signed vote.
    \item The proposer of vertex $v$ waits until it collects $n-f$ votes for $v$, aggregates them into a vertex certificate, and then broadcasts this certificate to all replicas. Finally, a replica formally incorporates $v$ into its local DAG once it receives either the certificate for $v$ or at least $2f+1$ votes for $v$, provided that all certificates referenced by $v$ are already present in its local storage.
\end{packedenumerate}

\begin{figure}[!t]
    \centering
    \includegraphics[width=0.65\linewidth]{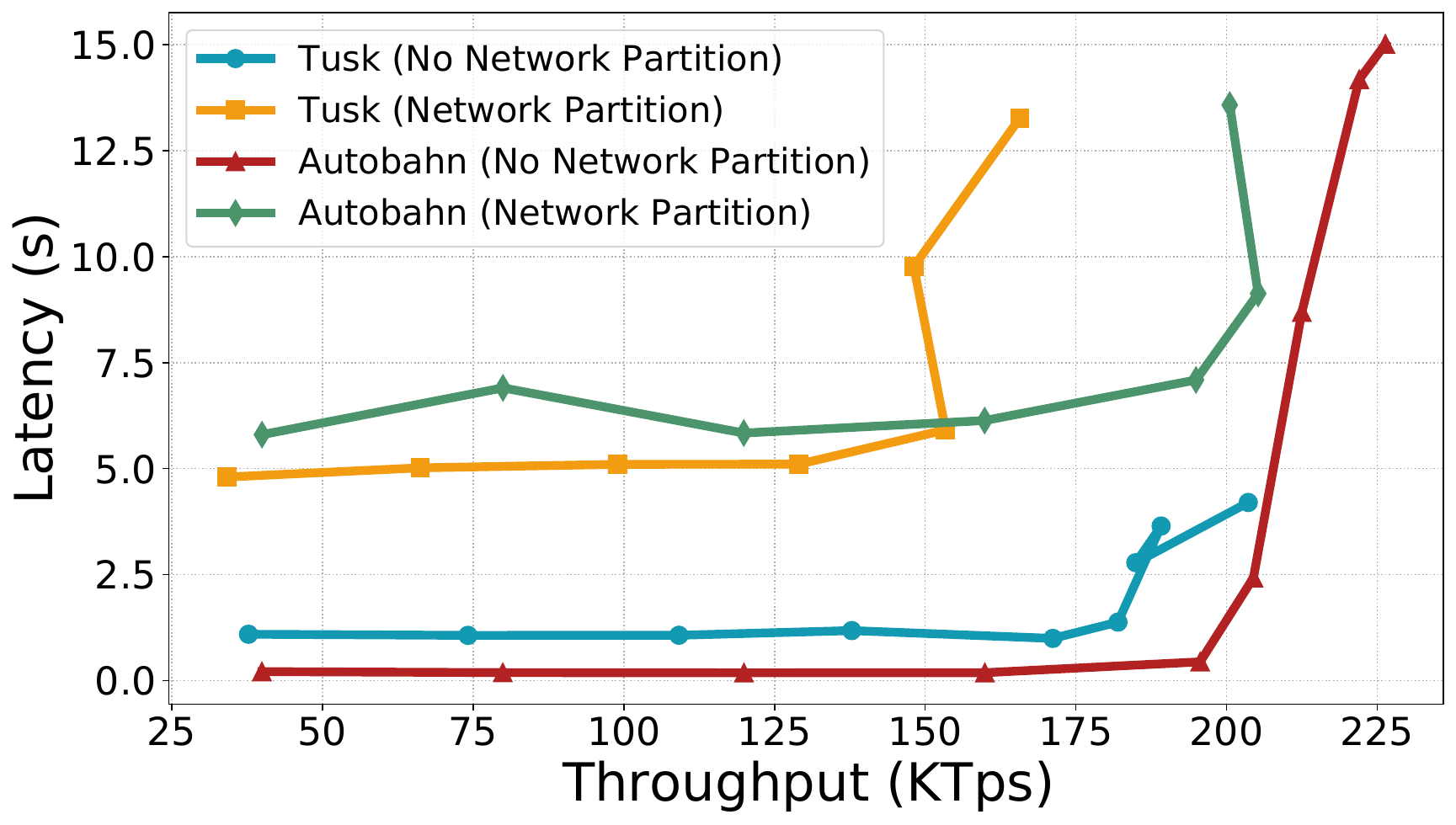}
    \caption{Asynchronous vs. partially synchronous DAG-based BFT under partition WAN conditions.}
    \label{fig:network_partition}
\end{figure}

\balance
\section{Performance under Heterogeneous Network} \label{app: Asyn_network}
Modern distributed systems, especially blockchains, operate over highly heterogeneous wide-area networks (WANs), where geographically distributed nodes experience diverse latencies, bandwidth asymmetries, and transient workload fluctuations. Prior studies~\cite{miller2016honey,guo2020dumbo} show that partially synchronous BFT protocols can suffer severe latency inflation or stalled progress when network conditions deviate from timely delivery assumptions. This challenge has motivated a line of asynchronous DAG-based BFT protocols~\cite{keidar2021all,huang2025chitu}, which seek stronger robustness while retaining high performance.

To further evaluate this robustness, we compare a representative asynchronous protocol (Tusk) with a partially synchronous protocol (Autobahn) under highly heterogeneous conditions. We inject large communication delays between selected node clusters to emulate partition-like WAN environments. Fig.~\ref{fig:network_partition} shows that Tusk exhibits a moderate throughput drop of about 13\% and a 4.5$\times$ increase in average latency. By contrast, Autobahn maintains similar throughput, but its latency rises far more sharply, by over 30$\times$ on average. These results highlight the stronger robustness of asynchronous designs under highly heterogeneous WAN conditions.

\end{document}